\documentclass[%
 aps,
 pra,
 reprint,
 amsmath,amssymb,
 superscriptaddress,
 longbibliography
]{revtex4-2}
\allowdisplaybreaks

\usepackage{graphicx}
\usepackage{dcolumn}
\usepackage{bm}
\usepackage{multirow}
\usepackage{amsfonts}
\usepackage{amsthm}
\usepackage{mathtools}
\usepackage{mathrsfs}
\usepackage{xcolor}
\usepackage{textcomp}
\usepackage{booktabs}
\usepackage{cancel}
\usepackage{placeins}
\usepackage{todonotes}
\def\added#1{#1}
\def\revnote#1{}

\theoremstyle{plain}
\newtheorem{lemma}{Lemma}

\theoremstyle{definition}

\theoremstyle{remark}

\usepackage{orcidlink}
\usepackage[T1]{fontenc}
\usepackage[utf8]{inputenc}

\newcommand{\qinc}{
Quantum Innovation Centre (Q.InC), \href{https://ror.org/036wvzt09}{Agency for Science, Technology and Research (A*STAR)}, 2 Fusionopolis Way, Innovis \#08-03, Singapore 138634, Republic of Singapore\looseness=-1}

\newcommand{\sutd}{Science, Mathematics and Technology Cluster, \href{https://ror.org/05j6fvn87}{Singapore University of Technology and Design}, 8 Somapah Road, Singapore 487372, Republic of Singapore\looseness=-1}

\newcommand{\ihpc}{\href{https://ror.org/02n0ejh50}{Institute of High Performance Computing (IHPC)}, \href{https://ror.org/036wvzt09}{Agency for Science, Technology and Research (A*STAR)}, 1 Fusionopolis Way, \#16-16 Connexis, Singapore 138632, Republic of Singapore\looseness=-1}

\newcommand{\sit}{Engineering Cluster, \href{https://ror.org/01v2c2791}{Singapore Institute of Technology}, 1 Punggol Coast Road, Singapore 828608, Republic of Singapore\looseness=-1}

\hypersetup{colorlinks=true,linkcolor=blue,citecolor=blue,urlcolor=blue}

\begin{document}

\title{Entanglement in the Quantum Volunteer's Dilemma}

\author{Noah Dane Hebdon\,\orcidlink{0000-0002-2855-0798}}
\email{nhebdon2@jh.edu}
\affiliation{\qinc}

\author{Dax Enshan Koh\,\orcidlink{0000-0002-8968-591X}}
\email{dax.koh@singaporetech.edu.sg}
\affiliation{\qinc}
\affiliation{\ihpc}
\affiliation{\sutd}
\affiliation{\sit}

\begin{abstract}
A well-known model in game theory, the Volunteer's Dilemma describes a group of $n$ players who decide whether to volunteer for a collective benefit at a personal cost, or to abstain and risk forfeiting the benefit altogether. A quantum version of this dilemma, developed within the Eisert–Wilkens–Lewenstein framework, allows each player to manipulate one qubit of a shared entangled state, leading to symmetric Nash equilibria with higher expected payoffs than in the classical game. Existing analyses, however, assume maximal entanglement\added{, which is experimentally demanding for hardware-constrained quantum systems}. Within the same framework, we introduce a generalized Quantum Volunteer's Dilemma with a tunable entanglement parameter $\gamma$ and study the extent to which equilibrium behavior depends on the level of entanglement. We derive explicit conditions relating $\gamma$, the number of players, and the players' strategies under which symmetric Nash equilibria exist, focusing on two canonical strategy profiles: one for $2\leq n\leq 9$, and one for even $n$. We find that maximal entanglement is not required to sustain symmetric equilibria. Instead, equilibrium behavior persists above a threshold value, which we compute analytically in both cases. We also demonstrate that the threshold value directly depends on system size.
\end{abstract}

\maketitle

\section{Introduction}\label{sec1}

Entanglement is a central resource that distinguishes quantum information processing from its classical counterpart \cite{horodecki2009quantum}. It underlies many important applications in quantum information science, including quantum computing \cite{ekert1998quantum,jozsa2003role,bruss2011multipartite,biamonte2020entanglement,diezvalle2021quantum, you2021exploring,katabarwa2022connecting}, quantum communication \cite{horodecki2001mixed,ursin2007entanglement,zou2021quantum,xing2023fundamental,xing2024teleportation}, quantum cryptography \cite{ekert1991quantum,jennewein2000quantum,yin2020entanglement}, and quantum metrology \cite{giovannetti2011advances, huang2016usefulness, augusiak2016asymptotic}. Entanglement has also proven to be a useful strategic asset in quantum game theory, where classical games are generalized within a quantum framework \cite{eisert199quantum}.

In quantum games, players choose quantum operations to apply to their respective quantum systems, rather than being restricted to classical pure or mixed strategies \cite{eisert2000quantum,flitney2002introduction,khan2018quantum}. Such formulations can introduce non-classical correlations among players’ strategies, change payoff structures and equilibrium conditions, and enable outcomes that are inaccessible in the corresponding classical game. Quantum game theory therefore provides a natural setting in which entanglement can be studied as a strategic resource.

Early work by Meyer showed that a player with access to quantum operations can outperform a player restricted to classical strategies in a simple finite game setting~\cite{meyer1999quantum}. Eisert, Wilkens, and Lewenstein then developed one of the central frameworks for quantizing nonzero-sum games, showing through the Prisoner’s Dilemma that quantum strategies can change the equilibrium structure of a classical dilemma~\cite{eisert199quantum}. In this framework, players apply local unitary operations to a shared quantum state, with payoffs assigned according to the outcome of a final measurement.

Since these early developments, quantum game theory has been applied to a broad range of game-theoretic scenarios, including multiplayer \cite{benjamin2001multi-player,tsakiroglou2026graphencoded,essalmi2026multi-player}, continuous-variable \cite{li2002continuous}, evolutionary \cite{kay2001evolutionary}, and other generalized settings \cite{sanzmartin2025mapping}, as well as to a variety of game-theoretic models \cite{du2000nash,piotrowski2002quantum,chen2004n,nawaz2004dilemma,frackiewicz2009ultimate,consuelo2020pareto,szopa2021efficiency,kastampolidou2023quantum, frackiewicz2024permissible,frackiewicz2025permissible_four,koh2025quantum,andronikos2025ghz,essalmi2025quantum,smolinski2026quantum}, with computational tools developed to support their analysis \cite{vlachos2009quantum,kotara2023software,grzanka2025qegs}. Studies of noise have examined how noise affects quantum-game behavior \cite{johnson2001playing,chen2003quantum,flitney2004quantum,shuai2007effect,huang2016quantum,khan2018dynamics,kairon2020noisy,legon2023joint,allah2026possibility}, whilst experiments across platforms have connected quantum games to practical issues of noise and control \cite{lu2004linear,buluta2006quantum,mitra2007experimental,prevedel2007experimental,schmid2010experimental,xu2022experimental,agreda2025bridging,agreda2026experimental}. 

In parallel, various works have investigated the role of quantum resources in quantum strategic advantage \cite{du2001entanglement,du2002entanglement,ozdemir2007necessary,nawaz2010quantum,nawaz2013strategic,li2014entanglement,wei2017quantum,santos2019entanglement,mohamed2023quantum,bugu2025entanglement,varsamis2025analysis}. Among these resources, entanglement plays a central role. Du \emph{et al.}\ showed that entanglement can induce threshold behavior in the equilibrium structure \cite{du2001entanglement}, while their later work showed that it can enhance equilibrium payoffs in multiplayer games \cite{du2002entanglement}. Li and Yong showed that entanglement can guarantee the emergence of cooperation in evolutionary quantum games on networks \cite{li2014entanglement}. More recently, Bugu showed that multipartite entanglement can enhance coordination and strategic efficiency in adversarial quantum games \cite{bugu2025entanglement}. Together, these results show that entanglement is not merely an auxiliary feature of quantum games, but a resource that can shape strategic behavior itself.

This perspective motivates the study of entanglement in multiplayer collective-action games. The Volunteer’s Dilemma is a canonical example of such a problem. First introduced by Diekmann \cite{diekmann1985volunteers}, it describes a situation in which individuals in a group decide whether to volunteer, thereby securing a collective benefit at a personal cost, or to abstain, thereby avoiding the cost but risking the loss of the collective benefit altogether. The dilemma captures a broad class of social traps in which individuals prefer that a public benefit be provided, but also prefer that someone else bear the cost. It has therefore been used to study leadership emergence~\cite{smith2016leadership}, burden-sharing~\cite{weesie1998cost}, free-riding~\cite{otsubo2008dynamic}, the bystander effect~\cite{latane1970unresponsive,campos-mercade2021volunteers}, and diffusion of responsibility~\cite{darley1968bystander,barron2002private}, with applications in economic modeling~\cite{battaglini2024welfare}, climate policy~\cite{hilbe2014cooperation}, and military strategy~\cite{amir2026volunteers}.
Its scope also extends beyond human decision-making, with applications to collective behavior in animal groups \cite{searcy2010evolution, mielke2019snake, steinegger2020laboratory, heifetz2021arabian,padget2026sharing}, cancer cells \cite{morsky2018cheater,archetti2019cooperation,manini2022ecology}, and microbial populations \cite{pedroso2018impact, patel2019crystal}. Several variants have also been introduced, including asymmetric \cite{diekmann1993cooperation,weesie1993asymmetry,he2014evolutionary,healy2018cost,guo2023asymmetric}, threshold \cite{chen2013shared,mago2023greed}, timing \cite{weesie1993asymmetry,weesie1994incomplete}, and cost-sharing versions \cite{weesie1998cost,amir2026repeated,amir2026volunteers}. In particular, the cost-sharing version introduced by Weesie and Franzen allows the burden of volunteering to be shared equally among multiple volunteers~\cite{weesie1998cost}.

The Quantum Volunteer's Dilemma of Koh, Kumar, and Goh \cite{koh2025quantum} builds on the Weesie--Franzen cost-sharing formulation~\cite{weesie1998cost}.
In the underlying classical game, the pure Nash equilibria are asymmetric: equilibrium occurs when exactly one player volunteers, while the remaining players receive the collective benefit without paying the cost. A symmetric mixed-strategy equilibrium also exists, but this restores symmetry only at the level of expected payoffs. In any particular realization of the game, there remains a nonzero probability that no player volunteers, which eliminates any collective benefit. Thus, the Classical Volunteer’s Dilemma contains both an asymmetry problem, because the burden falls on a single volunteer, and a coordination problem, because symmetric play can still lead to collective failure \cite{weesie1993asymmetry}.

The Quantum Volunteer’s Dilemma revisits this coordination problem by allowing players to act on a shared entangled quantum state \cite{koh2025quantum}. Koh, Kumar, and Goh showed that, under maximal entanglement, the quantum game admits symmetric Nash equilibria with higher expected payoffs than the classical mixed equilibrium~\cite{koh2025quantum}. However, existing analyses assume that the shared state is maximally entangled. This leaves open an important question: are these symmetric equilibria a feature only of the maximally entangled case, or do they persist when the amount of entanglement is reduced?

In this paper, we address this gap by introducing a variable entanglement parameter and studying how the equilibrium behavior changes as the amount of entanglement is varied. In doing so, we characterize the amount of entanglement required for the quantum game to retain the symmetric equilibrium behavior found in the maximally entangled case.

The remainder of the paper is organized as follows. In Sec.~\ref{sec2}, we introduce the mathematical formulations of the Classical Volunteer’s Dilemma and the Quantum Volunteer’s Dilemma. In Sec.~\ref{sec3}, we generalize the quantum game by introducing the entanglement parameter $\gamma$. In Secs.~\ref{sec4} and \ref{sec5}, we analyze two strategy sets of interest and derive lower bounds on $\gamma$ in terms of the number of players $n$. Finally, in Sec.~\ref{sec6}, we summarize our results.

\section{Game-Theoretic Formulation}~\label{sec2}
\vspace{-3em}
\subsection{The Classical Game}
We consider the cost-sharing formulation developed by Weesie and Franzen \cite{weesie1998cost}. Each of the $n\in \mathbb{Z}_{\geq 2}$ players may either abstain (0) or volunteer (1). In other words, each player $i$ has access to the following strategy set $T_i = \{0,1\}$ as they attempt to maximize their respective payoff $\$_i$.
Our $n$-player game is then defined as a $2n$-tuple
\begin{align}
    G = (T_1, T_2, \ldots, T_n;\, \$_1, \$_2, \ldots, \$_n),
\end{align}
where each $T_i$ is a set and each \[\$_i : T_1 \times T_2 \times \cdots \times T_n \to \mathbb{R}\]
is a real-valued function. We also define \[x = (x_1, x_2, \ldots, x_n)\in\{0,1\}^n\] as the game's strategy profile with the Hamming weight function producing the total number of volunteers: \[\mathrm{wt}(x) = \left|\{ i \in [n] : x_i = 1 \}\right|.\]

We assume that each player is rational and attempts to maximize their payoff based on the following conditions. If no players volunteer, then all players receive a payoff of zero. If $k > 0$ players volunteer, each
non-volunteer receives benefit $b$ while each volunteer receives $b - \frac{c}{k}$.
Throughout this work, we set $b = 2$ and $c = 1$ for simplicity. \added{We also use $\mathbb{I}[P]$ to denote the indicator of a proposition $P$, which is equal to $1$ if $P$ is true and $0$ otherwise.} The payoff function for player $i$ is therefore
\begin{align}
    \$_i^{\mathrm{VD}}(x)
    =
    \begin{cases}
        2 - \dfrac{1}{\mathrm{wt}(x)}, & x_i = 1, \\[6pt]
        2\,\added{\mathbb{I}}[\mathrm{wt}(x) > 0], & x_i = 0 .
    \end{cases}
    \label{eq:payoffFirst}
\end{align}

\subsection{Nash Equilibria}
To evaluate the quality of a strategy profile, we use the notion of a Nash equilibrium~\cite{nash1951non}, which describes a scenario where no player has an incentive to deviate from their strategy. A profile $x$ is ``stable'' in this sense if no player $i$ can improve their payoff by changing their own strategy while all others are held fixed:
\[
    \$_i(x) \ge \$_i(x \oplus e_i),
\]
where $e_i$ is the binary vector with a $1$ in the $i$-th position and $0$ elsewhere, and $\oplus$ denotes bitwise addition.
More generally, a strategy profile $s = (s_1, \ldots, s_n)$ is a Nash equilibrium of a game $G$ if, for every player $i \in [n]$ and every alternative strategy $t_i \in T_i \setminus \{s_i\}$,
\begin{align}
    \$_i(s_1, \ldots, s_i, \ldots, s_n) \ge \$_i(s_1, \ldots, t_i, \ldots, s_n).
\end{align}

In the classical Volunteer's Dilemma, the only Nash equilibria are those in which exactly one player volunteers, i.e., $\mathrm{wt}(x)=1$. These equilibria are inherently asymmetric: one player incurs the cost while the others free-ride, leading to unequal payoffs. They are also fragile in practice, since coordinating to produce a single volunteer is difficult.
Several variants attempt to address this. Mixed-strategy formulations, for example, assign a common probability of volunteering across players~\cite{diekmann1985volunteers,weesie1998cost}. This restores symmetry at the level of expected payoffs, but it does not eliminate the possibility that no players volunteer. The quantum formulation of the game, however, addresses the issue well.

\subsection{The Quantum Game}
In the Quantum Volunteer's Dilemma, introduced by Koh, Kumar, and Goh, entanglement enables a form of coordinated volunteering \cite{koh2025quantum}. \added{Within the Eisert–Wilkens–Lewenstein (EWL) quantization framework \cite{eisert199quantum}, the game admits symmetric Nash equilibria that generally yield higher expected payoffs than the classical mixed-strategy equilibrium. In this version, the shared entangled multi-qubit state is prepared using the entanglement operator
    \begin{align}
        J = \mathrm{e}^{- \mathrm{i}\frac{\pi}{4} \, Y^{\otimes n}}
        = \frac{1}{\sqrt{2}}\left(I - \mathrm{i} Y^{\otimes n}\right),
    \end{align}
    where $Y = \begin{pmatrix}
    0 & -\mathrm{i} \\ \mathrm{i} & 0
    \end{pmatrix}$ is the Pauli-$Y$ matrix. Applying $J$ to the initial state $\lvert 0\rangle^{\otimes n}$ gives
    \begin{align}
        J\lvert 0\rangle^{\otimes n}=\frac{1}{\sqrt{2}}\left(\lvert 0\rangle^{\otimes n}-\mathrm{i}^{\,n+1}\lvert 1\rangle^{\otimes n}\right).
    \end{align}}
\indent Instead of simply \added{deciding} to either volunteer or abstain, the players use quantum strategies to manipulate the entangled state. To that end, each player selects from the two-parameter family of unitaries $\{U(\theta,\phi) : \theta \in [0,4\pi), \ \phi \in [0,2\pi)\}$,
where
\begin{align}
    U(\theta,\phi) =
    \begin{pmatrix}
        \mathrm{e}^{\mathrm{i}\phi}\cos\left(\frac{\theta}{2}\right) & \sin\left(\frac{\theta}{2}\right) \\
        -\sin\left(\frac{\theta}{2}\right) & \mathrm{e}^{-\mathrm{i}\phi}\cos\left(\frac{\theta}{2}\right)
    \end{pmatrix}.
\end{align}
\added{This two-parameter family is the restricted EWL strategy set adopted in Ref.~\cite{koh2025quantum}. It can be written as
\[
    U(\theta,\phi)=R_z(-\phi)R_y(-\theta)R_z(-\phi),
\]
where $R_y(\alpha)=\mathrm{e}^{-\mathrm{i}\alpha \frac{Y}{2}}$ and $R_z(\alpha)=\mathrm{e}^{-\mathrm{i}\alpha \frac{Z}{2}}$. Thus, $\theta$ controls mixing through a $Y$ rotation, while $\phi$ controls relative phases through $Z$ rotations.}

\added{Each player's strategy set is represented by $\Theta:= [0,4\pi)\times[0,2\pi)$, with player $i$'s strategy given by $(\theta_i,\phi_i)\in\Theta$, corresponding to the operation $U(\theta_i,\phi_i)$. We define $\Theta^n$ as the $n$-fold Cartesian product of $\Theta$, with strategy profiles written as $
(\theta,\phi):=\bigl((\theta_1,\phi_1),\ldots,(\theta_n,\phi_n)\bigr)$.}

\added{All subsequent Nash-equilibrium results are relative to the choice of strategy set $\Theta$ and may depend on the class of allowed local operations. Once a strategy profile has been selected and the corresponding local operations have been applied, a final joint measurement determines which players volunteer and the payoff each player receives.} \revnote{"Key Definitions" moved from Sec.~III A to this position for clarity.}

\subsubsection{Key Definitions}
We now introduce several definitions developed by Koh, Kumar, and Goh~\cite{koh2025quantum} that will be used throughout the manuscript.

Let each player choose the parameters $\theta = (\theta_1, \ldots, \theta_n) \in \added{[0,4\pi)^n}$
and $\phi = (\phi_1, \ldots, \phi_n) \in [0,2\pi)^n$. This yields the following pre-measurement state:
\begin{align}
    \lvert \psi_f(\theta,\phi) \rangle
    =
    J^\dagger
    \left(
    \bigotimes_{i=1}^{n} U(\theta_i,\phi_i)
    \right)
    J \lvert 0 \rangle^{\otimes n}.
\end{align}

For a binary strategy profile $x \in \{0,1\}^n$, let
$\bar{x}$ denote its bitwise complement, with
$\bar{x}_i = 1 - x_i$ for each $i$. We can then define the amplitude and probability of observing the binary string $\bar{x}$ when the state $\lvert \psi_f \rangle$ is measured in the computational basis as $a_{\theta,\phi}(x)$ and $p_{\theta,\phi}(x)$, respectively. The amplitude and probability of measuring $\bar{x}$ are as follows:

\begin{align}
        a_{\theta,\phi}(x)
        &= \langle \bar{x} \mid \psi_f(\theta,\phi) \rangle \nonumber\\
        &= \langle \bar{x} | J^\dagger\cdot  \bigotimes_{i=1}^{n} U(\theta_i,\phi_i)\cdot
        J | 0\rangle^{\otimes n},
    \label{eq:amp1}
\end{align}
and
\begin{align}
        p_{\theta,\phi}(x)
        &= | a_{\theta,\phi}(x)|^2 \nonumber\\
        &= | \langle \bar{x} \mid \psi_f(\theta,\phi) \rangle|^2 \nonumber\\
        &= \left|\langle \bar{x} | J^\dagger \cdot
        \bigotimes_{i=1}^{n} U(\theta_i,\phi_i)\cdot
        J | 0\rangle^{\otimes n}\right|^2.
\end{align}
The expected payoff of player $i$ is given by the following:
\begin{align}
    \$_i(\theta,\phi)
    =
    \sum_{x \in \{0,1\}^n}
    \$_i^{\mathrm{VD}}(x)\, p_{\theta,\phi}(x),
\end{align}
where $\$_i^{\mathrm{VD}}$ represents the payoff function in the deterministic Volunteer's Dilemma as defined in Eq.~(\ref{eq:payoffFirst}). The payoff for player $i\in[n]$ is given by
\begin{align}
        \$_i(\theta,\phi)
        &= 2\sum_{\substack{x \in{\{0,1\}^n} \\ x\neq 0^n \\ x_i = 0}} p_{\theta,\phi}(x)+\sum_{k=1}^{n} \left(2-\frac{1}{k}\right)\sum_{\substack{x \in{\{0,1\}^n} \\ x_i=1 \\ \mathrm{wt}(x)=k}}p_{\theta,\phi}(x).
    \label{eq:paylong}
\end{align}
To simplify future operations, we also include the following definitions:
\begin{equation}
    \begin{aligned}
        c_{\theta, \phi}(x) = \cos\left({\sum_{i:x_i=1}\phi_i}\right)\prod_{i:x_i=0}\sin\left({\frac{\theta_i}{2}}\right)\prod_{i:x_i=1}\cos\left({\frac{\theta_i}{2}}\right),
    \end{aligned}
    \label{eq:c}
\end{equation}
\begin{equation}
    \begin{aligned}
        s_{\theta, \phi}(x) = \sin\left({\sum_{i:x_i=0}\phi_i}\right)\prod_{i:x_i=0}\cos\left({\frac{\theta_i}{2}}\right)\prod_{i:x_i=1}\sin\left({\frac{\theta_i}{2}}\right),
    \end{aligned}
    \label{eq:s}
\end{equation}
\begin{equation}
    \begin{aligned}
        w_{\theta, \phi}(x) = \sin\left({\sum_{i:x_i=1}\phi_i}\right)\prod_{i:x_i=0}\sin\left({\frac{\theta_i}{2}}\right)\prod_{i:x_i=1}\cos\left({\frac{\theta_i}{2}}\right).
    \end{aligned}
    \label{eq:w}
\end{equation}

Throughout their work~\cite{koh2025quantum}, Koh, Kumar, and Goh demonstrate two cases of symmetric Nash equilibria. The first concerns the strategy profile $Q^n=(0, \frac{\pi}{n})^n$, where a Nash equilibrium is shown to exist for $2\leq n\leq9$. The second concerns the strategy profile $A^n = (0, \frac{\pi}{2})^n$, where Nash equilibria exist for even $n$. In both cases, however, maximal entanglement is assumed, which is \added{experimentally demanding} for current quantum devices.

\added{Our analysis focuses specifically on these two previously identified symmetric strategy profiles and determines the entanglement levels for which they remain Nash equilibria. The resulting thresholds therefore do not exclude the existence of other symmetric equilibria, including potentially $\gamma$-dependent strategy profiles, below these thresholds.}

\section{Generalized Entanglement}\label{sec3}

\subsection{Variable Entanglement}
\added{
    To introduce variable entanglement, we generalize the entanglement operator to
    \begin{align}
        J^\gamma=\mathrm{e}^{-\mathrm{i}\frac{\gamma}{2}Y^{\otimes n}},
    \end{align}
    where $\gamma\in[0,\tfrac{\pi}{2}]$. The corresponding final state is
    \begin{align}
        \lvert\psi_f^\gamma(\theta,\phi)\rangle=(J^\gamma)^\dagger\left(\bigotimes_{i=1}^{n}U(\theta_i,\phi_i)\right)J^\gamma\lvert0\rangle^{\otimes n}.
    \end{align}
    This parametrization preserves the EWL entangling structure of Ref.~\cite{koh2025quantum} while continuously varying the entanglement of the resulting GHZ-class state. Across any nontrivial bipartition, the von Neumann entropy increases from $0$ at $\gamma=0$ to $1$ at $\gamma=\frac{\pi}{2}$.\\
    \indent If we define the fidelity between pure states as $F(\lvert\psi\rangle,\lvert\phi\rangle)=|\langle\psi\vert\phi\rangle|^2$, the fidelity of this state with the maximally entangled state at $\gamma=\frac{\pi}{2}$ is
    \begin{align}
        F\left(\lvert\psi(\gamma)\rangle,\lvert\psi(\tfrac{\pi}{2})\rangle\right)=\frac{1}{2}\left(1+\sin(\gamma)\right),
        \label{eq:fidelity}
    \end{align}
    which increases from $\frac{1}{2}$ at $\gamma=0$ to $1$ at $\gamma=\frac{\pi}{2}$. Other multipartite resource states, such as W states or graph states, as well as mixed states arising from decoherence, are outside the scope of this work. Whether the fidelity thresholds derived below extend to more general input states remains an open question.
}

\subsection{Probability Amplitude and Mass}
To allow for variable levels of entanglement in the Quantum Volunteer’s Dilemma, we first derive the generalized probability amplitude $a^\gamma_{\theta,\phi}(x)$ for observing the bit string $\bar{x}$ when the final state $\lvert \psi_f^\gamma \rangle$ is measured in the computational basis.
 We then define the probability distribution $p^\gamma_{\theta, \phi}(x)$. The proofs for the following lemmas can be found in Appendix~\ref{app:lemma1} and Appendix~\ref{app:lemma2}, respectively.

\begin{lemma}
\label{lemma:prob_amplitude}
The generalized probability amplitude associated with observing the bit string $\bar{x}$, for entanglement parameter $\gamma \in [0,\pi/2]$, is
\begin{align}
            a^\gamma_{\theta, \phi}(x) &= (-1)^{\mathrm{wt}(\bar{x})}c_{\theta, \phi}(x) + \mathrm{i}(-1)^{\mathrm{wt}(\bar{x})}\cos({\gamma})w_{\theta, \phi}(x)\nonumber\\
            &\quad-(\mathrm{i}^n)\sin({\gamma})s_{\theta, \phi}(x).
        \label{eq:amp2}
    \end{align}
\end{lemma}
\begin{lemma}
\label{lemma:prob}
The generalized probability of observing the bit string $\bar{x}$, with entanglement parameter $\gamma \in [0,\pi/2]$, is
        \begin{align}
            p_{\theta,\phi}^\gamma(x)&=c_{\theta,\phi}^2(x) + \sin^2(\gamma)s^2_{\theta, \phi}(x)+\cos^2(\gamma)w^2_{\theta, \phi}(x) \nonumber\\
            &\quad-s_{\theta, \phi}(x)(-1)^{\mathrm{wt}(x)} 
            \nonumber\\
            &\quad \times
            \begin{cases}
            2(-1)^{\frac{n}{2}}\sin(\gamma)c_{\theta, \phi}(x), & n\text{ even},\\
            (-1)^{\frac{n+1}{2}}\sin(2\gamma)w_{\theta, \phi}(x), & n\text{ odd}.
            \end{cases}
        \label{eq:prob1}
    \end{align}
\end{lemma}

\section{Symmetric Nash Equilibria for \texorpdfstring{$\mathbf{Q^n=(0, \frac{\pi}{n})^n}$}{Qn = (0, pi/n)n} Strategy Profile}\label{sec4}
We first consider the game $G^{(n)}_{\mathrm{QVD}}$ where each player employs the strategy $Q:=(0, \frac{\pi}{n})\in\Theta$. The strategy profile is then $Q^n=(Q,Q,\ldots ,Q)=((0, \frac{\pi}{n}),\ldots, (0, \frac{\pi}{n}))\in\Theta^n$. So, each player applies the following unitary operator:
    \begin{align}
        U\left(0,\frac{\pi}{n}\right)=
        \begin{pmatrix}
        \mathrm{e}^{\mathrm{i}\frac{\pi}{n}} & 0 \\
        0 & \mathrm{e}^{-\mathrm{i}\frac{\pi}{n}}
        \end{pmatrix}.
    \end{align}

We first show that by employing the $Q$ strategy, regardless of the level of entanglement, each player is rewarded with a payoff of $2-\frac{1}{n}$. We then show that a tight lower bound on $\gamma$ for the existence of Nash equilibria is given by

\begin{align}
        \gamma_{Q^n}^{\min}(n) = \arcsin\left(\frac{1}{\sqrt{n}\sin\left(\frac{\pi}{n}\right)}\right).
    \end{align}
\subsection{Payoff for Strategy Profile \texorpdfstring{$\mathbf{Q^n}$}{Qn}}
It has been shown previously that, in the case of maximal entanglement, every player volunteers with probability $1$ \cite{koh2025quantum}. In other words, the probability distribution is given by $p_{Q^n}(x) = \added{\mathbb{I}}[x = 1^n]$. The corresponding payoff of each player $i \in [n]$ has also been shown to be $2 - \frac{1}{n}$.

To show that this still holds in the general case with variable entanglement $\gamma$, we evaluate Eq.~\eqref{eq:c}, Eq.~\eqref{eq:s}, and Eq.~\eqref{eq:w} for this strategy profile. In Koh, Kumar, and Goh's work~\cite{koh2025quantum}, $c_{Q^n}(x)$ and $s_{Q^n}(x)$ are defined as Eq.~(42) and Eq.~(43), respectively. They are evaluated as $c_{Q^n}(x)=-\added{\mathbb{I}}[x=1^n]$, and $s_{Q^n}(x)=0$. We evaluate $w_{Q^n}(x)$ in a similar fashion in Appendix~\ref{app:q1} and find that $w_{Q^n}(x)=0$.
Hence,
\begin{align}
    p_{Q^n}^{\gamma}(x)=c_{Q^n}^2(x)=\left(-\added{\mathbb{I}}[x=1^n]\right)^2=\added{\mathbb{I}}[x=1^n].
\end{align}

Since $p_{Q^n}^{\gamma}(1^n)=1$, every player volunteers with probability 1, just as in the case with maximal entanglement. Evaluating Eq.~(\ref{eq:paylong}) with this value indicates that the generalized payoff with $\gamma$ is the same as the payoff where $\gamma=\frac{\pi}{2}$, which was originally evaluated in Eq.~(46) of Koh, Kumar, and Goh's work \cite{koh2025quantum}:
\begin{align}
        \$_i(Q^n)=\sum_{k=1}^{n}\left(2-\frac{1}{k}\right)\delta_{k,n}=2-\frac{1}{n}.
    \label{eq:payshortq}
\end{align}

\subsection{Lower Bound on Entanglement for Strategy Profile \texorpdfstring{$\mathbf{Q^n}$}{Qn}}
We now understand that all players will receive a payoff of $2-\frac{1}{n}$ if they adhere to strategy $Q$. This allows us to identify a condition on $n$ and $\gamma$ for the strategy profile $Q^n$ to be a Nash equilibrium. To that end, we consider the case where one player $a$ deviates from strategy $Q$ and determine the conditions under which player $a$ has access to a greater payoff than the $2-\frac{1}{n}$ upper bound. Any instance where player $a$'s payoff is greater than $2-\frac{1}{n}$ is not a Nash equilibrium. Note that all other players $i\neq a$ maintain strategy $Q$ and receive the payoff $2-\frac{1}{n}$.

We consider $c_{\theta, \phi}(x)$, $s_{\theta, \phi}(x)$, and $w_{\theta, \phi}(x)$ in this scenario. Koh, Kumar, and Goh previously developed formulas for $c_{\theta, \phi}(x)$ and $s_{\theta, \phi}(x)$ under these conditions in Eq.~(47) and Eq.~(48) of~\cite{koh2025quantum}, respectively. We similarly develop $w_{\theta, \phi}(x)$ in Appendix~\ref{app:q2} to get the following:
\begin{widetext}
\begin{align}
    c_{\theta, \phi}(x)\bigg|_{(\theta_i, \phi_i)=Q, \forall i\neq a}&=
        \begin{cases}
            -\cos\left({\frac{\pi}{n}}\right)\sin\left({\frac{\theta_a}{2}}\right), &\text{ for } x_a=0, x_i=1, \forall i \neq a,\\
            -\cos\left({\frac{\pi}{n}-\phi_a}\right)\cos\left({\frac{\theta_a}{2}}\right), &\text{ for } x_a=1, x_i=1, \forall i \neq a,\\
            0, &\text{ otherwise},
        \end{cases}\nonumber\\
        s_{\theta, \phi}(x)\bigg|_{(\theta_i, \phi_i)=Q, \forall i\neq a}&=
        \begin{cases}
            \sin\left({\frac{\pi}{n}-\phi_a}\right)\cos\left({\frac{\theta_a}{2}}\right), &\text{ for } x_a=0, x_i=0, \forall i \neq a,\\
            \sin\left({\frac{\pi}{n}}\right)\sin\left({\frac{\theta_a}{2}}\right), &\text{ for } x_a=1, x_i=0, \forall i \neq a,\\
            0, &\text{ otherwise},
        \end{cases}\nonumber\\
        w_{\theta, \phi}(x)\bigg|_{(\theta_i, \phi_i)=Q, \forall i\neq a}&=
        \begin{cases}
            \sin\left({\frac{\pi}{n}}\right)\sin\left({\frac{\theta_a}{2}}\right), &\text{ for } x_a=0, x_i=1, \forall i \neq a,\\
            \sin\left({\frac{\pi}{n}-\phi_a}\right)\cos\left({\frac{\theta_a}{2}}\right), &\text{ for } x_a=1, x_i=1, \forall i \neq a,\\
            0, &\text{ otherwise}.
        \end{cases}
\end{align}
It is important to note that there are disjoint supports between $c_{\theta, \phi}(x)$ and $s_{\theta, \phi}(x)$. In a similar way, there are disjoint supports between $w_{\theta, \phi}(x)$ and $s_{\theta, \phi}(x)$. Hence, each of their products is zero, which then gives

\begin{align}
        p^\gamma_{\theta, \phi}(x)\bigg|_{(\theta_i, \phi_i)=Q, \forall i\neq a}
        &=\cos^2\left({\frac{\pi}{n}-x_a\phi_a}\right)\sin^{2(1-x_a)}\left({\frac{\theta_a}{2}}\right)\cos^{2x_a}\left({\frac{\theta_a}{2}}\right)\added{\mathbb{I}}[\forall i \neq a, x_i=1]\nonumber\\
        &\quad + \cos^2({\gamma})\sin^2\left({\frac{\pi}{n}-x_a\phi_a}\right)\sin^{2(1-x_a)}\left({\frac{\theta_a}{2}}\right)\cos^{2x_a}\left({\frac{\theta_a}{2}}\right)\added{\mathbb{I}}[\forall i \neq a, x_i=1]\nonumber\\
        &\quad + \sin^2({\gamma})\sin^2\left({\frac{\pi}{n}-\bar{x}_a\phi_a}\right)\cos^{2(1-x_a)}\left({\frac{\theta_a}{2}}\right)\sin^{2x_a}\left({\frac{\theta_a}{2}}\right)\added{\mathbb{I}}[\forall i \neq a, x_i=0]\nonumber\\
        &=\begin{cases}
            \cos^2\left(\frac{\pi}{n}\right)\sin^2\left(\frac{\theta_a}{2}\right)
            + \cos^2(\gamma)\sin^2\left(\frac{\pi}{n}\right)\sin^2\left(\frac{\theta_a}{2}\right),
            & \text{ for } x_a=0,\ x_i=1,\forall i\neq a,\\ 
            \cos^2\left(\frac{\pi}{n}-\phi_a\right)\cos^2\left(\frac{\theta_a}{2}\right)
            + \cos^2(\gamma)\sin^2\left(\frac{\pi}{n}-\phi_a\right)\cos^2\left(\frac{\theta_a}{2}\right),
            & \text{ for } x_a=1,\ x_i=1,\forall i\neq a,\\
        \sin^2(\gamma)\sin^2\left(\frac{\pi}{n}-\phi_a\right)\cos^2\left(\frac{\theta_a}{2}\right),
            & \text{ for } x_a=0,\ x_i=0,\forall i\neq a,\\
\sin^2(\gamma)\sin^2\left(\frac{\pi}{n}\right)\sin^2\left(\frac{\theta_a}{2}\right),
            & \text{ for } x_a=1,\ x_i=0,\forall i\neq a,\\
            0, & \text{otherwise.}
        \end{cases}
    \end{align}
Above, the first case implies that $x=11\ldots1\underset{\uparrow_a}{0}11\ldots1$. Similarly, the second case implies that $x=1^n$. The third implies that $x=0^n$, and finally, the fourth case implies that $x=00\ldots0\underset{\uparrow_a}{1}00\ldots0$. With this in mind, we now substitute this result into Eq.~(\ref{eq:paylong}), which yields

    \begin{align}   &\$_a^\gamma(\theta,\phi)\bigg|_{(\theta_i, \phi_i)=Q, \forall i\neq a}
        = 2\sum_{\substack{x \in{\{0,1\}^n} \\ x\neq 0^n \\ x_a = 0}} p_{\theta,\phi}(x)
        + \sum_{k=1}^{n} \left(2-\frac{1}{k}\right)\sum_{\substack{x \in{\{0,1\}^n} \\ x_a=1 \\ \mathrm{wt}(x)=k}}p_{\theta,\phi}(x)\bigg|_{(\theta_i, \phi_i)=Q, \forall i\neq a} \nonumber\\
        &= 2p_{\theta,\phi}(1\ldots101\ldots1)+\left(2-\frac{1}{1}\right)p_{\theta,\phi}(0\ldots010\ldots0)+\left({2-\frac{1}{n}}\right)p_{\theta,\phi}(1^n)\bigg|_{(\theta_i, \phi_i)=Q, \forall i\neq a}
        \nonumber\\
        &=2\left( \cos^2\left({\frac{\pi}{n}}\right)\sin^2\left({\frac{\theta_a}{2}}\right)+\cos^2({\gamma})\sin^2\left({\frac{\theta_a}{2}}\right)\sin^2\left({\frac{\pi}{n}}\right)\right) + \sin^2({\gamma})\sin^2\left({\frac{\pi}{n}}\right)\sin^2\left({\frac{\theta_a}{2}}\right)\nonumber\\
        &\quad+ \left(2 - \frac{1}{n}\right)\left(\cos^2\left({\frac{\pi}{n}-\phi_a}\right)\cos^2\left({\frac{\theta_a}{2}}\right)+\cos^2(\gamma)\sin^2\left({\frac{\pi}{n}-\phi_a}\right)\cos^2\left({\frac{\theta_a}{2}}\right)\right).
\end{align}
\end{widetext}
\noindent To simplify this formula, we make the following substitutions:

\begin{align}
        f &= \sin^2\left({\frac{\theta_a}{2}}\right), &\ g &=\sin^2\left({\frac{\pi}{n}}\right),\nonumber\\
        h &= \sin^2\left(\gamma\right), &\
        k &= \sin^2\left({\frac{\pi}{n}}-\phi_a\right).
    \end{align}
Rearranging the terms, we then have
\begin{align}
        &\$_a^\gamma(\theta,\phi)\bigg|_{(\theta_i, \phi_i)=Q, \forall i\neq a} \nonumber\\
        &= 2f\left(\cos^2\left({\frac{\pi}{n}}\right)+g\cos^2(\gamma)\right)+fgh+\left(2-\frac{1}{n}\right) \nonumber\\
        &\quad\times\cos^2\left({\frac{\theta_a}{2}}\right)\left(\cos^2\left({\frac{\pi}{n}-\phi_a}\right)\added{+}k\cos^2(\gamma)\right)\nonumber\\
        &= 2f(1-hg) + fgh + \left(2-\frac{1}{n}\right)\cos^2\left({\frac{\theta_a}{2}}\right)(1-hk)\nonumber\\
        &= 2f(1-hg) + fgh +\left(2 - \frac{1}{n}\right)(1-f)(1-hk)\nonumber\\
        &= 2f(1-hg) + fgh + \left(2-\frac{1}{n}\right)(1-f)\nonumber\\
        &\quad -\left(2-\frac{1}{n}\right)(1-f)hk.
    \end{align}
We now compare this value to the maximum payoff that player $a$ receives by adhering to strategy $Q$, which is $2-\frac{1}{n}$.
For the current strategy to be a Nash equilibrium, it must be the case that $\$_a^\gamma(x) \leq 2-\frac{1}{n}$. We now simplify this relation to derive the condition for which a Nash equilibrium can exist within this strategy profile, beginning with
    \begin{align*}
    2-\frac{1}{n}
        &\geq
        2f(1-hg)+fgh+\left(2-\frac{1}{n}\right) \nonumber\\
        &\quad\times(1-f)
        -\left(2-\frac{1}{n}\right)(1-f)hk.
    \end{align*}
Rearranging and simplifying, we produce what follows:
    \begin{align}
        \left(2-\frac{1}{n}\right)(1-f)hk
        &\geq
        2f(1-hg)+fgh-\left(2-\frac{1}{n}\right)f,\nonumber\\
        \left(2-\frac{1}{n}\right)(1-f)hk
        &\geq
        f\left(2(1-hg)+gh-\left(2-\frac{1}{n}\right)\right),\nonumber\\
        \left(2-\frac{1}{n}\right)(1-f)hk
        &\geq
        f\left(2-2hg+hg-2+\frac{1}{n}\right),\nonumber\\
        \left(2-\frac{1}{n}\right)(1-f)hk
        &\geq
        f\left(\frac{1}{n}-hg\right).
    \end{align}
Note that in all cases, the \added{left-hand} side is nonnegative. For the inequality to hold, it must be the case that $\left(\frac{1}{n}-hg\right)\leq0$ or $\frac{1}{n} \leq \sin^2(\gamma)\sin^2\left(\frac{\pi}{n}\right)$.
Rearranging, we then get the following:
\[
\gamma \geq \arcsin\left(\frac{1}{\sqrt{n}\sin\left(\frac{\pi}{n}\right)}\right).
\]
Thus, the lower bound is
    \begin{align}
        \gamma_{Q^n}^{\min}(n) = \arcsin\left(\frac{1}{\sqrt{n}\sin\left(\frac{\pi}{n}\right)}\right).
    \label{eq:lowerbound1}
    \end{align}
\added{Utilizing Eq.~(\ref{eq:fidelity}), we discover that this threshold corresponds to the following fidelity bound:
\begin{align}
    F\left(\lvert\psi(\gamma)\rangle,\lvert\psi(\tfrac{\pi}{2})\rangle\right)\geq\frac{1}{2}\left(1+\frac{1}{\sqrt{n}\sin(\tfrac{\pi}{n})}\right).
\end{align}
}

With \added{Eq.~(\ref{eq:lowerbound1})} satisfied, the strategy profile $Q^n$ represents a Nash equilibrium for all $\Theta_a$ and $\Phi_a$. Because $\gamma \in [0, \frac{\pi}{2}]$, it is evident that $n \in [2,9]$, which is the range of $n$ for which Nash equilibria were known to exist in the case of maximal entanglement. \added{It is also clear that varying $\gamma$ does not extend the range of $n$ for which the $Q^n$ strategy profile is a Nash equilibrium. This behavior is illustrated in Fig.~\ref{fig:plot1}. As the number of players increases, the range of entanglement values $\gamma$ that preserve the Nash equilibrium initially widens slightly before steadily narrowing. Remarkably, for $n=3$, the equilibrium persists at entanglement levels slightly below $\gamma = \frac{\pi}{4}$. These results demonstrate that maximal entanglement is not required to sustain the symmetric Nash equilibrium associated with the $Q^n$ strategy profile.}

\begin{figure}[t]
    \centering
    \includegraphics[width=0.85\linewidth]{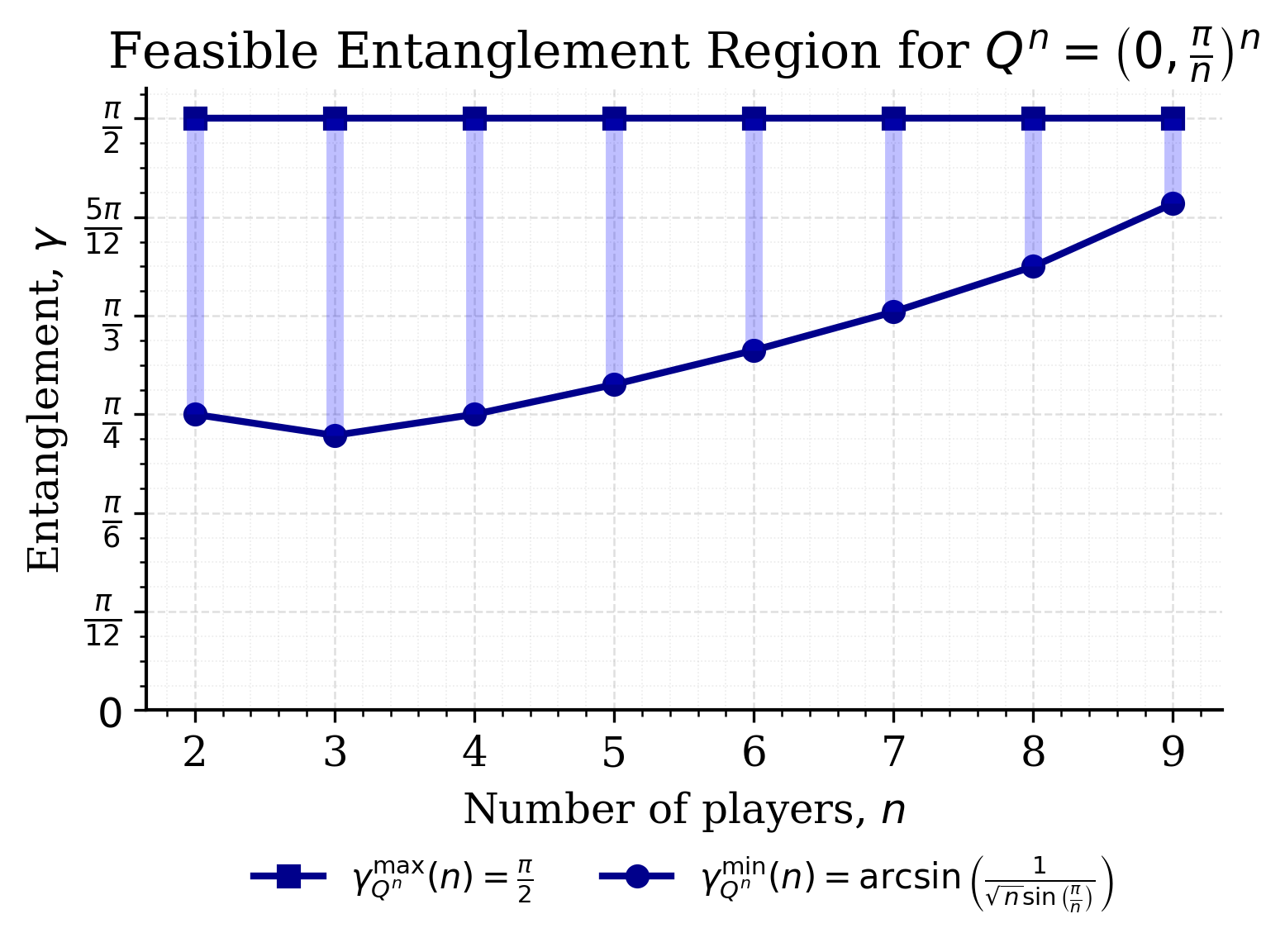}
    \caption{Lower bound on the entanglement parameter $\gamma$ as a function of the number of players $n$ for the $Q^n = \left(0, \frac{\pi}{n}\right)^n$ strategy profile. The shaded regions along discrete values of $n$ indicate that $\gamma \geq \gamma_{\min}(n)$ is satisfied, which is the condition for Nash equilibrium.}
    \label{fig:plot1}
\end{figure}

\section{Symmetric Nash Equilibria for Strategy Profile \texorpdfstring{$\mathbf{A^n=(0, \frac{\pi}{2})^n}$}{An = (0, pi/2)n}}\label{sec5}
We now consider the case where all players choose the strategy $A=\left(0, \frac{\pi}{2}\right)\in\Theta$. This results in the following strategy profile: $A^n=(A,\ldots,A)\in\Theta^n$, where each player applies the following unitary operator:
\begin{equation}
        U\left(0, \frac{\pi}{2}\right) = \begin{pmatrix}
            \mathrm{i} &0\\
            0 &-\mathrm{i}
        \end{pmatrix}=\mathrm{i}Z,
\end{equation}
where $Z = \begin{pmatrix}
            1 &0\\
            0 & -1
        \end{pmatrix}$ is the Pauli-$Z$ matrix.
We now determine the payoff of each player for the $A^n$ strategy profile. As will be demonstrated, the probability function naturally separates into the cases where $n$ is even and where $n$ is odd. We show that, regardless of the level of entanglement, \added{the strategy profile $A^n$ is not a Nash equilibrium when n is odd}. We then show that a tight lower bound on $\gamma$ for the existence of symmetric Nash equilibria for even $n$ is given by
\begin{align}
        \added{\gamma_{A^n}^{\min}(n)=\arcsin\left(\frac{1}{\sqrt{n}}\right).}
\end{align}
\medskip

\subsection{Payoff for Strategy Profile \texorpdfstring{$\mathbf{A^n}$}{An}}
To determine the payoff of each player for the $A^n$ strategy profile, we refer to the evaluation of $c_{\theta, \phi}(x)$ and $s_{\theta, \phi}(x)$ in Eq.~(53) and Eq.~(54) of Koh, Kumar, and Goh's work \cite{koh2025quantum}. In Appendix~\ref{app:a1}, we evaluate $w_{\theta, \phi}(x)$ in a similar way to get
    \begin{align}
        c_{A^n}(x)&= (-1)^{\frac{n}{2}}\added{\mathbb{I}}[n \text{ even}]\added{\mathbb{I}}[x=1^n], \nonumber\\
        s_{A^n}(x) &= (-1)^{\frac{n-1}{2}}\added{\mathbb{I}}[n \text{ odd}]\added{\mathbb{I}}[x = 0^n], \nonumber\\
        w_{A^n}(x)&=(-1)^{\frac{n-1}{2}}\added{\mathbb{I}}[n \text{ odd}]\added{\mathbb{I}}[x=1^n].
    \end{align}

We now apply these to Eq.~(\ref{eq:prob1}) and call it $p^\gamma_{A^n}(x)$. There are two cases to consider.

\paragraph{Case 1: $n$ is even.}
In this case, it follows that $s_{A^n}(x)=w_{A^n}(x)=0$, which yields
    \begin{align}
        p^\gamma_{A^n}(x) &= c_{A^n}^2(x) \nonumber\\
        &= \left((-1)^{\frac{n}{2}}\added{\mathbb{I}}[n \text{ even}]\added{\mathbb{I}}[x=1^n]\right)^2 \nonumber\\
        &=\added{\mathbb{I}}[n \text{ even}]\added{\mathbb{I}}[x=1^n].
    \end{align}

\paragraph{Case 2: $n$ is odd.}
When $n$ is odd, it follows that $c_{A^n}(x)=0$, which produces the following:
    \begin{align}
        p^\gamma_{A^n}(x)
        &= \sin^2(\gamma)s_{A^n}^2(x)
            + \cos^2(\gamma)w_{A^n}^2(x) \nonumber\\
        &\quad - s_{A^n}(x)(-1)^{\mathrm{wt}(x)}
            (-1)^{\frac{n+1}{2}}\sin(2\gamma)w_{A^n}(x) \nonumber\\
        &= \sin^2(\gamma)s_{A^n}^2(x)
            + \cos^2(\gamma)w_{A^n}^2(x) \nonumber\\
        &\quad- (-1)^{\frac{2\mathrm{wt}(x)+n+1}{2}}
            s_{A^n}(x)w_{A^n}(x)\sin(2\gamma) \nonumber\\
        &= \sin^2(\gamma)
            \left((-1)^{\frac{n-1}{2}}\added{\mathbb{I}}[n \text{ odd}]\added{\mathbb{I}}[x=0^n]\right)^2 \nonumber\\
        &\quad+ \cos^2(\gamma)
            \left((-1)^{\frac{n-1}{2}}\added{\mathbb{I}}[n \text{ odd}]\added{\mathbb{I}}[x=1^n]\right)^2 \nonumber\\
        &\quad- (-1)^{\frac{2\mathrm{wt}(x)+n+1}{2}}
            \left((-1)^{\frac{n-1}{2}}\added{\mathbb{I}}[n \text{ odd}]\added{\mathbb{I}}[x=0^n]\right) \nonumber\\
        &\quad\times
            \left((-1)^{\frac{n-1}{2}}\added{\mathbb{I}}[n \text{ odd}]\added{\mathbb{I}}[x=1^n]\right)
            \sin(2\gamma) \nonumber\\
        &= \added{\mathbb{I}}[n \text{ odd}]
            \left(
                \sin^2(\gamma)\added{\mathbb{I}}[x=0^n]+ \cos^2(\gamma)\added{\mathbb{I}}[x=1^n]
            \right).
    \end{align}
Notice that the cross term vanished on the final step, since $\added{\mathbb{I}}[x=0^n]\added{\mathbb{I}}[x=1^n]=0$. Also, note that when $n=0$, it follows that $x=0^n$ and $x=1^n$. However, $n\geq 2$. Hence,
    \begin{align}
        p^\gamma_{A^n}(x)&= \added{\mathbb{I}}[n \text{ even}]\added{\mathbb{I}}[x=1^n]+\added{\mathbb{I}}[n \text{ odd}] \nonumber\\
        &\quad\times\left(\sin^2(\gamma)\added{\mathbb{I}}[x=0^n]+\cos^2(\gamma)\added{\mathbb{I}}[x=1^n]\right).
    \end{align}
Next, we substitute this back into Eq.~(\ref{eq:paylong}), which produces the following:
\begin{widetext}
\begin{align}
    \$_i^\gamma(A^n)
        &=2\sum_{\substack{x \in{\{0,1\}^n} \nonumber\\ x\neq 0^n \\ x_i = 0}} p^\gamma_{A^n}(x)
        + \sum_{k=1}^{n} \left(2-\frac{1}{k}\right)\sum_{\substack{x \in{\{0,1\}^n} \\ x_i=1 \\ \mathrm{wt}(x)=k}}p^\gamma_{A^n}(x)
        \nonumber\\
        &= 2\sum_{\substack{x \in{\{0,1\}^n} \\ x\neq 0^n \\ x_i = 0}}\left(\added{\mathbb{I}}[n \text{ even}]\added{\mathbb{I}}[x=1^n]+\added{\mathbb{I}}[n \text{ odd}]\left(\added{\mathbb{I}}[x=0^n]\sin^2(\gamma)+\added{\mathbb{I}}[x=1^n]\cos^2(\gamma)\right)\right)
        \nonumber\\
        &\quad +\sum_{k=1}^{n} \left(2-\frac{1}{k}\right)\sum_{\substack{x \in{\{0,1\}^n} \\ x_i=1 \\ \mathrm{wt}(x)=k}}\left(\added{\mathbb{I}}[n \text{ even}]\added{\mathbb{I}}[x=1^n]+\added{\mathbb{I}}[n \text{ odd}]\left(\added{\mathbb{I}}[x=0^n]\sin^2(\gamma)+\added{\mathbb{I}}[x=1^n]\cos^2(\gamma)\right)\right)
        \nonumber\\
        &=\sum_{k=1}^{n} \left(2-\frac{1}{k}\right)\sum_{\substack{x \in{\{0,1\}^n} \\ x_i=1 \\ \mathrm{wt}(x)=k}}\left(\added{\mathbb{I}}[n \text{ even}]\added{\mathbb{I}}[x=1^n]+\added{\mathbb{I}}[n \text{ odd}]\added{\mathbb{I}}[x=1^n]\cos^2(\gamma)\right)
        \nonumber\\
        &=\sum_{k=1}^{n} \left(2-\frac{1}{k}\right)\sum_{\substack{x \in{\{0,1\}^n} \\ x_i=1 \\ \mathrm{wt}(x)=k}}\added{\mathbb{I}}[x=1^n]\left(\added{\mathbb{I}}[n \text{ even}]+\added{\mathbb{I}}[n \text{ odd}]\cos^2(\gamma)\right)
        \nonumber\\
        &= \sum_{k=1}^{n} \left(2-\frac{1}{k}\right)\delta_{k,n}\left(\added{\mathbb{I}}[n \text{ even}]+\added{\mathbb{I}}[n \text{ odd}]\cos^2(\gamma)\right)
        \nonumber\\
        &= \left(2-\frac{1}{n}\right)\left(\added{\mathbb{I}}[n \text{ even}]+\added{\mathbb{I}}[n \text{ odd}]\cos^2(\gamma)\right).
    \label{eq:paylonga}
\end{align}
\end{widetext}

\subsection{Lower Bound on Entanglement for Strategy Profile \texorpdfstring{$\mathbf{A^n}$}{An}}
With the payoff defined for players who employ $A$ in the case of an even and an odd number of players, we now allow one player $a\in[n]$ to deviate from the strategy set $A$ to determine which conditions produce a greater payoff for player $a$. Once again, any payoff greater than $\$^\gamma(A^n)$ describes a state of non-Nash equilibrium.

To develop a lower bound on $\gamma$ in terms of $n$, we refer to the evaluation of $c_{\theta, \phi}(x)$ and $s_{\theta, \phi}(x)$ in Eq.~(56) and Eq.~(59) of Koh, Kumar, and Goh's \added{published} work \cite{koh2025quantum}. We then evaluate $w_{\theta, \phi}(x)$ by a similar method. For the derivation, refer to Appendix~\ref{app:a2}. We have
\begin{widetext}
    \begin{align}
        c_{\theta, \phi}(x)\bigg|_{(\theta_i,\phi_i)=A, \forall i\neq a}&=\left((-1)^{\frac{n}{2}}\sin(x_a\phi_a)\added{\mathbb{I}}[n \text{ even}]+(-1)^{\frac{n-1}{2}}\cos(x_a\phi_a)\added{\mathbb{I}}[n \text{ odd}]\right)
        \nonumber\\
        &\quad\times\sin^{1-x_a}\left(\frac{\theta_a}{2}\right)\cos^{x_a}\left(\frac{\theta_a}{2}\right)\added{\mathbb{I}}[\forall i\neq a, x_i=1],
        \nonumber\\
        s_{\theta, \phi}(x)\bigg|_{(\theta_i,\phi_i)=A, \forall i\neq a}&=\left(-(-1)^{\frac{n}{2}}\cos(\bar{x}_a\phi_a)\added{\mathbb{I}}[n \text{ even}]+(-1)^{\frac{n-1}{2}}\sin(\bar{x}_a\phi_a)\added{\mathbb{I}}[n \text{ odd}]\right)
        \nonumber\\
        &\quad\times\cos^{1-x_a}\left(\frac{\theta_a}{2}\right)\sin^{x_a}\left(\frac{\theta_a}{2}\right)\added{\mathbb{I}}[\forall i\neq a, x_i=0],
        \nonumber\\
        w_{\theta, \phi}(x)\bigg|_{(\theta_i,\phi_i)=A, \forall i\neq a}&=\left(-(-1)^{\frac{n}{2}}\cos(x_a\phi_a)\added{\mathbb{I}}[n \text{ even}]+(-1)^{\frac{n-1}{2}}\sin(x_a\phi_a)\added{\mathbb{I}}[n \text{ odd}]\right)
        \nonumber\\
        &\quad\times\sin^{1-x_a}\left(\frac{\theta_a}{2}\right)\cos^{x_a}\left(\frac{\theta_a}{2}\right)\added{\mathbb{I}}[\forall i\neq a, x_i=1].
    \end{align}
Take note that $c_{\theta, \phi}(x)$ and $s_{\theta, \phi}(x)$ are nonzero on disjoint sets of $x$. Similarly, $s_{\theta,\phi}(x)$ and $w_{\theta, \phi}(x)$ are nonzero on disjoint sets of $x$. We enter these values into Eq.~(\ref{eq:prob1}), noting that the last term disappears:
\begin{align}
        p_{\theta, \phi}^{\gamma}(x)&\bigg|_{(\theta_i,\phi_i)=A, \forall i\neq a}=\left(c_{\theta, \phi}^2(x) + \sin^2(\gamma)s_{\theta, \phi}^2(x)+\cos^2(\gamma)w_{\theta, \phi}^2(x)\right)\bigg|_{(\theta_i,\phi_i)=A, \forall i\neq a}
        \nonumber\\
        &= \left(\sin^2(x_a\phi_a)\added{\mathbb{I}}[n \text{ even}]+\cos^2(x_a\phi_a)\added{\mathbb{I}}[n\text{ odd}]\right)\sin^{2(1-x_a)}\left(\frac{\theta_a}{2}\right)\cos^{2x_a}\left(\frac{\theta_a}{2}\right)\added{\mathbb{I}}[\forall i\neq a, x_i=1]
        \nonumber\\
        &\quad+ \sin^2(\gamma)\left(\cos^2(\bar{x}_a \phi_a)\added{\mathbb{I}}[n \text{ even}]+\sin^2(\bar{x}_a\phi_a)\added{\mathbb{I}}[n \text{ odd}]\right)\cos^{2(1-x_a)}\left(\frac{\theta_a}{2}\right)\sin^{2x_a}\left(\frac{\theta_a}{2}\right)\added{\mathbb{I}}[\forall i\neq a, x_i=0]
        \nonumber\\
        &\quad+ \cos^2(\gamma)\left(\cos^2(x_a\phi_a)\added{\mathbb{I}}[n \text{ even}]+\sin^2(x_a\phi_a)\added{\mathbb{I}}[n \text{ odd}]\right)\sin^{2(1-x_a)}\left(\frac{\theta_a}{2}\right)\cos^{2x_a}\left(\frac{\theta_a}{2}\right)\added{\mathbb{I}}[\forall i\neq a, x_i=1].
    \end{align}
We now separately consider cases where $n$ is odd and $n$ is even. For both cases, we first compute $p_{\theta, \phi}^{\gamma}(x)$ where $(\theta_i,\phi_i)=A,$ and $\forall i\neq a$. Then, we plug each result into Eq.~(\ref{eq:paylong}). We begin with the case where $n$ is odd.

\noindent\paragraph{Case 1: $n$ is odd.} In this case, we have
\begin{align}
        p_{\theta, \phi}^{\gamma}(x)\bigg|_{(\theta_i,\phi_i)=A, \forall i\neq a}
        &= \cos^2(x_a\phi_a)\sin^{2(1-x_a)}\left(\frac{\theta_a}{2}\right)\cos^{2x_a}\left(\frac{\theta_a}{2}\right)\added{\mathbb{I}}[\forall i\neq a, x_i=1]
        \nonumber\\
        &\quad+ \sin^2(\gamma)\sin^2(\bar{x}_a\phi_a)\cos^{2(1-x_a)}\left(\frac{\theta_a}{2}\right)\sin^{2x_a}\left(\frac{\theta_a}{2}\right)\added{\mathbb{I}}[\forall i\neq a, x_i=0]
        \nonumber\\
        &\quad+ \cos^2(\gamma)\sin^2(x_a\phi_a)\sin^{2(1-x_a)}\left(\frac{\theta_a}{2}\right)\cos^{2x_a}\left(\frac{\theta_a}{2}\right)\added{\mathbb{I}}[\forall i\neq a, x_i=1]
        \nonumber\\
        &=
        \begin{cases}
        \sin^2(\gamma)\sin^2(\phi_a)\cos^2\left(\frac{\theta_a}{2}\right), &\text{ for } x=0^n,\\
        \cos^2(\phi_a)\cos^2\left(\frac{\theta_a}{2}\right)+\cos^2(\gamma)\sin^2(\phi_a)\cos^2\left(\frac{\theta_a}{2}\right), &\text{ for } x=1^n,\\
        \sin^2\left(\frac{\theta_a}{2}\right), &\text{ for } x_a=0, x_i=1, \forall i\neq a,\\
        0, &\text{ otherwise,}
        \end{cases}
    \end{align}
which yields
\begin{align}
\$_a^\gamma(\theta,\phi)\bigg|_{(\theta_i, \phi_i)=A, \forall i\neq a}
        &= \left.\left(2\sum_{\substack{x \in{\{0,1\}^n} \\ x\neq 0^n \\ x_a = 0}} p_{A^n}^\gamma(x)
        + \sum_{k=1}^{n} \left(2-\frac{1}{k}\right)\sum_{\substack{x \in{\{0,1\}^n} \\ x_a=1 \\ \mathrm{wt}(x)=k}}p_{A^n}^\gamma(x)\right)\right|_{(\theta_i, \phi_i)=A, \forall i\neq a}
        \nonumber\\
        &=2\sin^2\left(\frac{\theta_a}{2}\right)+\left(2-\frac{1}{n}\right)\left(\cos^2(\phi_a)\cos^2\left(\frac{\theta_a}{2}\right)+\cos^2(\gamma)\sin^2(\phi_a)\cos^2\left(\frac{\theta_a}{2}\right)\right).
    \end{align}
\end{widetext}

Recall from Eq.~(\ref{eq:paylonga}) that $\$^\gamma_i(A^n)=\left(2-\frac{1}{n}\right)\cos^2(\gamma)$.
With this in mind, we now maximize $\$^\gamma_a(\theta, \phi)$ to determine the conditions under which it is bounded above by $\$^\gamma_i(A^n)$.
We factor the second term of $\$^\gamma_a(\theta, \phi)$ to get
\begin{align}
        \$^\gamma_a&(\theta, \phi)\bigg|_{(\theta_i,\phi_i)=A, \forall i\neq a}=2\sin^2\left(\frac{\theta_a}{2}\right)+\left(2-\frac{1}{n}\right)\nonumber\\
        &\quad\times\cos^2\left(\frac{\theta_a}{2}\right)(\cos^2(\phi_a)+\cos^2(\gamma)\sin^2(\phi_a)).
    \end{align}
First, we maximize with respect to $\phi_a$. Notice that 
\begin{align*}
\cos^2(\phi_a)+\cos^2(\gamma)\sin^2(\phi_a)
    &=1-\sin^2(\phi_a)(1-\cos^2(\gamma))\\
    &=1-\sin^2(\phi_a)\sin^2(\gamma).
\end{align*}
This value is maximized when $\phi_a=0$ or $\pi$, in which case $\cos^2(0)+\cos^2(\gamma)\sin^2(0)=1$. Under these conditions, we then have

\begin{align}
        \$^\gamma_a(\theta, \phi)\bigg|_{(\theta_i,\phi_i)=A, \forall i\neq a}
        &=2\sin^2\left(\frac{\theta_a}{2}\right) \nonumber\\
        &\quad+\left(2-\frac{1}{n}\right)\cos^2\left(\frac{\theta_a}{2}\right).
    \end{align}
Then, maximizing with respect to $\theta_a$ gives us $\theta_a=\pi$ and $\$_a^\gamma(\theta, \phi)=2$.
Hence, $\$_a^\gamma(\theta, \phi)$ is maximized at a value of 2 when $\phi_a=0$ or $\pi$, and $\theta_a=\pi$.
Because this maximum value of 2 is greater than $\$_a^\gamma(A^n)=\left(2-\frac{1}{n}\right)\cos^2(\gamma)$ for all values of $\gamma$, the collective strategy profile $A^n$ is not a Nash equilibrium when $n$ is odd, and leveraging entanglement cannot change that fact.

\noindent\paragraph{Case 2: $n$ is even.}
When $n$ is even, the probability and payoff functions evaluate to

\begin{widetext}
    \begin{align}
        p^\gamma_{\theta, \phi}(x)\bigg|_{(\theta_i,\phi_i)=A, \forall i\neq a}&=\sin^2(x_a\phi_a)\sin^{2(1-x_a)}\left(\frac{\theta_a}{2}\right)\cos^{2x_a}\left(\frac{\theta_a}{2}\right)\added{\mathbb{I}}[\forall i\neq a, x_i=1]\nonumber\\
        &\quad+\sin^2(\gamma)\cos^2(\bar{x}_a\phi_a)\cos^{2(1-x_a)}\left(\frac{\theta_a}{2}\right)\sin^{2x_a}\left(\frac{\theta_a}{2}\right)\added{\mathbb{I}}[\forall i\neq a, x_i=0]\nonumber\\
        &\quad+\cos^2(\gamma)\cos^2(x_a\phi_a)\sin^{2(1-x_a)}\left(\frac{\theta_a}{2}\right)\cos^{2x_a}\left(\frac{\theta_a}{2}\right)\added{\mathbb{I}}[\forall i\neq a, x_i=1]\nonumber\\
        &=\begin{cases}
            \sin^2(\gamma)\cos^2(\phi_a)\cos^2\left(\frac{\theta_a}{2}\right), &\text{ for } x=0^n,\\
            \sin^2(\phi_a)\cos^2\left(\frac{\theta_a}{2}\right)+\cos^2(\gamma)\cos^2(\phi_a)\cos^2\left(\frac{\theta_a}{2}\right), &\text{ for } x=1^n,\\
            \sin^2(\gamma)\sin^2\left(\frac{\theta_a}{2}\right), &\text{ for } x_a=1, x_i=0, \forall i\neq a,\\
            \cos^2(\gamma)\sin^2\left(\frac{\theta_a}{2}\right), &\text{ for } x_a=0, x_i=1, \forall i\neq a,\\
            0, &\text{ otherwise,}
        \end{cases}
    \end{align}
and
\begin{equation}
    \begin{aligned}
        \$^\gamma_a(\theta, \phi)\bigg|_{(\theta_i,\phi_i)=A, \forall i\neq a}&=\added{\left(\sin^2(\gamma)+2\cos^2(\gamma)\right)}\sin^2\left(\frac{\theta_a}{2}\right)\\
        &\quad+\left(2-\frac{1}{n}\right)\left(\sin^2(\phi_a)\cos^2\left(\frac{\theta_a}{2}\right)+\cos^2(\gamma)\cos^2(\phi_a)\cos^2\left(\frac{\theta_a}{2}\right)\right).
    \end{aligned}    
\end{equation}
\end{widetext}
Recall from Eq.~(\ref{eq:paylonga}) that $\$^\gamma_a(A^n)=2-\frac{1}{n}$. Just as before, we now maximize $\$^\gamma_a(\theta, \phi)$ to determine the conditions under which it is bounded above by $\$^\gamma_a(A^n)$.
To simplify the computation, we make the following substitutions:
\begin{align}
        t&=\cos^2(\gamma), &u=\left(2-\frac{1}{n}\right), \nonumber\\
        v&=\sin\left(\frac{\theta_a}{2}\right),&y =\cos\left(\frac{\theta_a}{2}\right).
    \end{align}
Substituting these into the payoff function, we have
\begin{align}
        \$^\gamma_a(\theta, \phi)\bigg|_{(\theta_i,\phi_i)=A, \forall i\neq a} &= \added{(1+t)v^2}+uy^2 \nonumber\\
        &\quad\times\left(\sin^2(\phi_a)+t\cos^2(\phi_a)\right).
\end{align}
We first maximize with respect to $\phi_a$. Note that $\sin^2(\phi_a)+t\cos^2(\phi_a)=t+(1-t)\sin^2(\phi_a)$. Since $0\leq t\leq 1$, the payoff is maximized when $\sin^2(\phi_a)=1$, or equivalently when $\phi_a\equiv\frac{\pi}{2}\pmod{\pi}$. Next, following the process for the previous case, we maximize with respect to $\theta_a$.
We have
\begin{align}
        \$^\gamma_a\left(\theta, \phi\right)\bigg|_{(\theta_i,\phi_i)=A, \forall i\neq a}&=\added{(1+t)v^2+uy^2} \nonumber\\
        &=\added{(1+t)(1-y^2)+uy^2} \nonumber\\
        &=\added{1+t+(u-t-1)y^2}.
    \end{align}
\added{If $u-t-1<0$, the expression is maximized at $y^2=0$, which corresponds to $\theta_a=\pi$. The resulting payoff is then $1+t$. Since $u-t-1<0$ is equivalent to $1+t>u$, this constitutes a profitable deviation. Hence, $A^n$ is a Nash equilibrium if and only if $u-t-1\geq0$. Equivalently, $A^n$ is a Nash equilibrium precisely when}
\[
\added{
\gamma \geq \arccos\left(\sqrt{1-\frac{1}{n}}\right)=\arcsin\left(\frac{1}{\sqrt{n}}\right).
}
\]
\added{This gives the lower bound}
\begin{equation}
    \added{
    \gamma_{A^n}^{\min}(n)=\arcsin\left(\frac{1}{\sqrt{n}}\right),
    }
\end{equation}
\added{with the corresponding fidelity bound
\begin{align}
    F\left(\lvert\psi(\gamma)\rangle,\lvert\psi(\tfrac{\pi}{2})\rangle\right)\geq\frac{1}{2}\left(1+\frac{1}{\sqrt{n}}\right).
\end{align}
}

This result is shown in Fig.~\ref{fig:plot2}, which demonstrates again that maximal entanglement is not required to realize symmetric Nash equilibria in the Quantum Volunteer's Dilemma. It is also interesting to note that, for this strategy profile, increasing the number of players increases the range of acceptable values of $\gamma$. The figure also demonstrates that the system is not constrained to have $n\leq9$ players as it was with the $Q^n$ strategy profile. The system is constrained, however, to even values of $n$.
\begin{figure}[!htbp]
    \centering
    \includegraphics[width=0.85\linewidth]{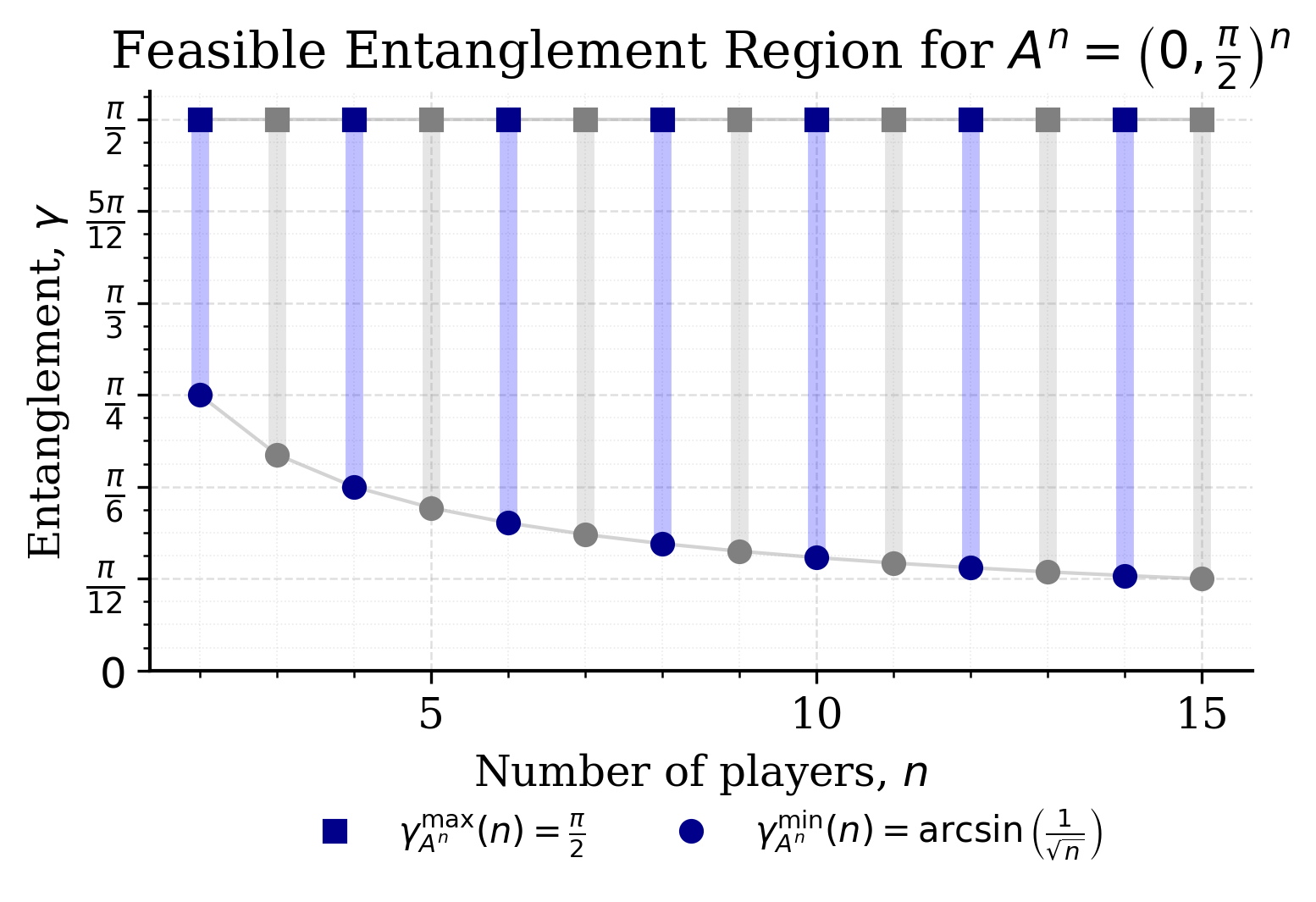}
    \caption{Lower bound on the entanglement parameter $\gamma$ as a function of the number of players $n$ for the $A^n = \left(0, \frac{\pi}{2}\right)^n$ strategy profile. The blue shaded lines along discrete values of even $n$ indicate that $\gamma \geq \gamma_{\min}(n)$.}
    \label{fig:plot2}
\end{figure}

\FloatBarrier

\section{Conclusion}\revnote{Figure 2 updated to reflect the corrected lower bound.}
\label{sec6}
In this study, we introduced a generalization of the Quantum Volunteer's Dilemma by incorporating the entanglement parameter $\gamma$, thereby extending the EWL-based formulation of Koh, Kumar, and Goh~\cite{koh2025quantum}. This generalization allowed us to analyze the extent to which the previously identified symmetric Nash equilibria depend on maximal entanglement to exist. In particular, we studied two families of symmetric Nash equilibria.

The first relates to the strategy profile $Q^n=(0,\frac{\pi}{n})^n$ while the second relates to $A^n=(0,\frac{\pi}{2})^n$. In both cases, we showed that the profile remains a Nash equilibrium for suitable non-maximal values of $\gamma$, provided that a specific inequality is satisfied. These conditions quantify how the entanglement threshold for each symmetric Nash equilibrium depends on the number of players $n$.

In summary, our results show that the strategic advantages of the Quantum Volunteer's Dilemma are not confined to the case of maximal entanglement. \added{Symmetric Nash equilibria can persist under partial entanglement, suggesting a degree of robustness to reduced entanglement within the idealized pure-state family considered here. Extending this analysis to decoherence and mixed-state noise is an important and natural direction for future work.}
\begin{acknowledgments}
This research is supported by the Agency for Science, Technology and Research (A*STAR) under the Quantum Innovation Centre (Q.InC) Strategic Research and Translational Thrust (SRTT).
\end{acknowledgments}

\appendix
\onecolumngrid
\section{Proof of Lemma \texorpdfstring{\ref{lemma:prob_amplitude}}{1}---Deriving the Generalized Probability Amplitude}
\label{app:lemma1}
\noindent\textit{Proof of Lemma~\ref{lemma:prob_amplitude}.} \
We begin with the entanglement operator $J^\gamma$:
    \begin{align}
        J^\gamma = \mathrm{e}^{- \mathrm{i}\frac{\gamma}{2} \, Y^{\otimes n}}
        = \cos\left(\frac{\gamma}{2}\right) - \mathrm{i} \sin\left(\frac{\gamma}{2}\right)Y^{\otimes n}.
    \end{align}
The initial state that the players first receive is then
    \begin{align}
        J^\gamma \lvert 0\rangle^{\otimes n}
        &=\left(\cos\left({\frac{\gamma}{2}}\right)I - \mathrm{i}\sin\left(\frac{\gamma}{2}\right)Y^{\otimes n}\right)\lvert 0\rangle^{\otimes n} \nonumber\\
        &= \cos\left(\frac{\gamma}{2}\right)\lvert 0\rangle^{\otimes n} - \mathrm{i}\sin\left(\frac{\gamma}{2}\right)Y^{\otimes n}\lvert 0\rangle^{\otimes n} \nonumber\\
        &=\cos\left(\frac{\gamma}{2}\right)\lvert 0\rangle^{\otimes n} - (\mathrm{i}^{n+1})\sin\left(\frac{\gamma}{2}\right)\lvert 1\rangle^{\otimes n}.
    \end{align}
Allowing all players to apply local operations renders
    \begin{align}
        \bigotimes_{i=1}^{n} U(\theta_i,\phi_i)\cdot J^\gamma \lvert
        0\rangle^{\otimes n}
        &=\bigotimes_{i=1}^{n} U(\theta_i,\phi_i)\cdot\left(\cos\left(\frac{\gamma}{2}\right)\lvert 0\rangle^{\otimes n} - (\mathrm{i}^{n+1})\sin\left(\frac{\gamma}{2}\right)\lvert 1\rangle^{\otimes n}\right)\nonumber\\
        &= \bigotimes_{i=1}^{n} U(\theta_i,\phi_i)\cdot\cos\left(\frac{\gamma}{2}\right)\lvert 0\rangle^{\otimes n} - (\mathrm{i}^{n+1})\bigotimes_{i=1}^{n} U(\theta_i,\phi_i)\cdot\sin\left(\frac{\gamma}{2}\right)\lvert 1\rangle^{\otimes n}.
    \label{eq:amp11}
\end{align}
The elements of the measurement basis are then rewritten as follows:
    \begin{align}
        J^\gamma \lvert \bar{x}\rangle
        &= \cos\left(\frac{\gamma}{2}\right)\lvert \bar{x}\rangle-\mathrm{i}\sin\left(\frac{\gamma}{2}\right)\bigotimes_{i=1}^{n}Y\lvert \bar{x}_i\rangle= \cos\left(\frac{\gamma}{2}\right)\lvert \bar{x}\rangle-\left(\mathrm{i}^{n+1}\right)(-1)^{\mathrm{wt}(\bar{x})}\sin\left(\frac{\gamma}{2}\right)\lvert{x}\rangle.
    \label{eq:amp12}
\end{align}
We now expand Eq.~\eqref{eq:amp1} with Eq.~\eqref{eq:amp11} and Eq.~\eqref{eq:amp12}:
    \begin{align}
        a^\gamma_{\theta,\phi}(x)
        &= \langle \bar{x} | J^{\gamma\dagger}\cdot  \bigotimes_{i=1}^{n} U(\theta_i,\phi_i)\cdot J^\gamma | 0\rangle^{\otimes n}= \left(\cos\left(\frac{\gamma}{2}\right)\langle \bar{x}\lvert-(\mathrm{i}^{-(n+1)})(-1)^{\mathrm{wt}(\bar{x})}\sin\left(\frac{\gamma}{2}\right)\langle{x}\vert\right)
        \nonumber\\
        &\quad\cdot \left(\bigotimes_{i=1}^{n} U(\theta_i,\phi_i)\cos\left(\frac{\gamma}{2}\right)\lvert 0\rangle - \left(\mathrm{i}^{n+1}\right)\bigotimes_{i=1}^{n} U(\theta_i,\phi_i)\sin\left(\frac{\gamma}{2}\right)\lvert 1\rangle\right) 
        \nonumber\\
        &= \cos^2\left(\frac{\gamma}{2}\right)\left(\prod_{i=1}^n\langle \bar{x}_i \rvert    U(\theta_i,\phi_i) \lvert 0\rangle\right)- \left(\mathrm{i}^{n+1}\right)\sin\left(\frac{\gamma}{2}\right)\cos\left(\frac{\gamma}{2}\right)\left(\prod_{i=1}^n\langle \bar{x}_i \rvert U(\theta_i,\phi_i)\lvert 1\rangle\right)
        \nonumber\\
        &\quad- \left(\mathrm{i}^{-(n+1)}\right)(-1)^{\mathrm{wt}(\bar{x})}\sin\left(\frac{\gamma}{2}\right)\cos\left(\frac{\gamma}{2}\right)\left(\prod_{i=1}^n\langle x_i \rvert U(\theta_i,\phi_i)\lvert 0\rangle\right)
        \nonumber\\
        &\quad+ (-1)^{\mathrm{wt}(\bar{x})}\sin^2\left(\frac{\gamma}{2}\right) \left(\prod_{i=1}^n\langle x_i \rvert U(\theta_i,\phi_i)\lvert 1\rangle\right).
    \end{align}
We then change the index from $\bar{x}$ to $x$.
Since $\bar{x}_i = 1 - x_i$, note that
    \begin{align}
        \prod_{i}\langle \bar{x}_i \rvert U_i\lvert 0\rangle&=\prod_{i:x_i=0}\langle1\rvert U_i\lvert 0\rangle \prod_{i:x_i=1}\langle0 \rvert U_i\lvert 0\rangle,
        \nonumber\\
        \prod_{i}\langle \bar{x}_i \rvert U_i\lvert 1\rangle&= \prod_{i:x_i=0}\langle1\rvert U_i\lvert 1\rangle \prod_{i:x_i=1}\langle0 \rvert U_i\lvert 1\rangle,
        \nonumber\\
        \prod_{i}\langle x_i \rvert U_i\lvert 0\rangle&=\prod_{i:x_i=0}\langle0 \rvert U_i\lvert 0\rangle \prod_{i:x_i=1}\langle1\rvert U_i\lvert 0\rangle, 
        \nonumber\\
        \prod_{i}\langle x_i \rvert U_i\lvert 1\rangle&=\prod_{i:x_i=0}\langle0 \rvert U_i\lvert 1\rangle \prod_{i:x_i=1}\langle1\rvert U_i\lvert 1\rangle.
    \end{align}
Applying this to the amplitude equation, we have
    \begin{align}
        a^\gamma_{\theta,\phi}(x)
        &= \cos^2\left(\frac{\gamma}{2}\right) \prod_{i:x_i=0}\langle1\rvert U_i\lvert 0\rangle\prod_{i:x_i=1}\langle0 \rvert U_i\lvert 0\rangle -\left(\mathrm{i}^{n+1}\right)\sin\left(\frac{\gamma}{2}\right)\cos\left(\frac{\gamma}{2}\right) \prod_{i:x_i=0}\langle1\rvert U_i\lvert 1\rangle \prod_{i:x_i=1}\langle0 \rvert U_i\lvert 1\rangle \nonumber\\
        &\quad- \left(\mathrm{i}^{-(n+1)}\right)(-1)^{\mathrm{wt}(\bar{x})}\sin\left(\frac{\gamma}{2}\right)\cos\left(\frac{\gamma}{2}\right)\prod_{i:x_i=0}\langle0 \rvert U_i\lvert 0\rangle \prod_{i:x_i=1}\langle1\rvert U_i\lvert 0\rangle \nonumber\\
        &\quad+ (-1)^{\mathrm{wt}(\bar{x})}\sin^2\left(\frac{\gamma}{2}\right) \prod_{i:x_i=0}\langle0 \rvert U_i\lvert 1\rangle \prod_{i:x_i=1}\langle1\rvert U_i\lvert 1\rangle \nonumber\\
        &= \cos^2\left(\frac{\gamma}{2}\right) \prod_{i:x_i=0}\left(-\sin\left(\frac{\theta_i}{2}\right)\right) \prod_{i:x_i=1}\left(\mathrm{e}^{\mathrm{i}\phi_i}\cos\left(\frac{\theta_i}{2}\right)\right) \nonumber\\
        &\quad-\left(\mathrm{i}^{n+1}\right)\sin\left(\frac{\gamma}{2}\right)\cos\left(\frac{\gamma}{2}\right) \prod_{\substack{i:x_i=0}}\left(\mathrm{e}^{-\mathrm{i} \phi_i} \cos\left(\frac{\theta_i}{2}\right)\right) \prod_{\substack{i:x_i=1}}\sin\left(\frac{\theta_i}{2}\right)
        \nonumber\\
        &\quad- \left(\mathrm{i}^{-(n+1)}\right)(-1)^{\mathrm{wt}(\bar{x})}\sin\left(\frac{\gamma}{2}\right)\cos\left(\frac{\gamma}{2}\right)\prod_{\substack{i:x_i=0}}\left(\mathrm{e}^{\mathrm{i} \phi_i}\cos\left(\frac{\theta_i}{2}\right)\right) \prod_{\substack{i:x_i=1}}\left(-\sin\left(\frac{\theta_i}{2}\right)\right)
        \nonumber\\
        &\quad+ (-1)^{\mathrm{wt}(\bar{x})}\sin^2\left(\frac{\gamma}{2}\right) \prod_{\substack{i:x_i=0}}\left(\sin\left(\frac{\theta_i}{2}\right)\right) \prod_{\substack{i:x_i=1}}\left(\mathrm{e}^{-\mathrm{i} \phi_i}\cos\left(\frac{\theta_i}{2}\right)\right)\nonumber\\
        &= (-1)^{\mathrm{wt}(\bar{x})}\cos^2\left(\frac{\gamma}{2}\right)\mathrm{e}^{\mathrm{i} \sum_{(i : x_i = 1)} \phi_i}\prod_{i:x_i=0}\sin{\left(\frac{\theta_i}{2}\right)}\prod_{i:x_i=1}\cos{\left(\frac{\theta_i}{2}\right)}\nonumber\\
        &\quad-\left(\mathrm{i}^{n+1}\right)\sin\left(\frac{\gamma}{2}\right)\cos\left(\frac{\gamma}{2}\right)\mathrm{e}^{-\mathrm{i} \sum_{(i : x_i = 0)} \phi_i}\prod_{i:x_i=0}\cos{\left(\frac{\theta_i}{2}\right)}\prod_{i:x_i=1}\sin{\left(\frac{\theta_i}{2}\right)}\nonumber\\
        &\quad-\left(-\mathrm{i}^{n+1}\right)(-1)^{\mathrm{wt}(\bar{x})}(-1)^{\mathrm{wt}(x)}\sin\left(\frac{\gamma}{2}\right)\cos\left(\frac{\gamma}{2}\right)\mathrm{e}^{\mathrm{i} \sum_{(i : x_i = 0)} \phi_i}\prod_{i:x_i=0}\cos{\left(\frac{\theta_i}{2}\right)}\prod_{i:x_i=1}\sin{\left(\frac{\theta_i}{2}\right)}\nonumber\\
        &\quad+ (-1)^{\mathrm{wt}(\bar{x})}\sin^2\left(\frac{\gamma}{2}\right)\mathrm{e}^{-\mathrm{i} \sum_{(i : x_i = 1)} \phi_i}\prod_{i:x_i=0}\sin{\left(\frac{\theta_i}{2}\right)}\prod_{i:x_i=1}\cos{\left(\frac{\theta_i}{2}\right)}.
    \end{align}
Note that, in the first term, we factored $(-1)^{\mathrm{wt}(\bar{x})}$ from $\prod_{(i:x_i=0)}\left(-\sin{\left(\frac{\theta_i}{2}\right)}\right)$. Similarly, we factored $(-1)^{\mathrm{wt}(x)}$ from $\prod_{(i:x_i=1)}\left(-\sin{\left(\frac{\theta_i}{2}\right)}\right)$.
Further simplifying, we get
    \begin{align}
        a^\gamma_{\theta,\phi}(x)
        &=(-1)^{\mathrm{wt}(\bar{x})}\left(\cos^2\left(\frac{\gamma}{2}\right)\mathrm{e}^{\mathrm{i} \sum_{(i : x_i = 1)} \phi_i} + \sin^2\left(\frac{\gamma}{2}\right)\mathrm{e}^{-\mathrm{i} \sum_{(i : x_i = 1)}\phi_i}\right)\prod_{i:x_i=0}\sin\left({\frac{\theta_i}{2}}\right)\prod_{i:x_i=1}\cos\left({\frac{\theta_i}{2}}\right)\nonumber\\
        &\quad- \left(\mathrm{i}^{n+1}\right)\sin\left(\frac{\gamma}{2}\right)\cos\left(\frac{\gamma}{2}\right)\mathrm{e}^{-\mathrm{i} \sum_{(i : x_i = 0)}\phi_i}\prod_{i:x_i=0}\cos\left({\frac{\theta_i}{2}}\right)\prod_{i:x_i=1}\sin\left({\frac{\theta_i}{2}}\right)\nonumber\\
        &\quad+ \mathrm{i}^{n+1}\sin\left(\frac{\gamma}{2}\right)\cos\left(\frac{\gamma}{2}\right)\mathrm{e}^{\mathrm{i} \sum_{(i : x_i = 0)}\phi_i}\prod_{i:x_i=0}\cos\left({\frac{\theta_i}{2}}\right)\prod_{i:x_i=1}\sin\left({\frac{\theta_i}{2}}\right)\nonumber\\
        &=(-1)^{\mathrm{wt}(\bar{x})}\left(\cos^2\left(\frac{\gamma}{2}\right)\mathrm{e}^{\mathrm{i} \sum_{(i : x_i = 1)} \phi_i} + \sin^2\left(\frac{\gamma}{2}\right)\mathrm{e}^{-\mathrm{i} \sum_{(i : x_i = 1)}\phi_i}\right)\prod_{i:x_i=0}\sin\left({\frac{\theta_i}{2}}\right)\prod_{i:x_i=1}\cos\left({\frac{\theta_i}{2}}\right)\nonumber\\
        &\quad+ \mathrm{i}^{n+1}\sin\left(\frac{\gamma}{2}\right)\cos\left(\frac{\gamma}{2}\right)\left(\mathrm{e}^{\mathrm{i} \sum_{(i : x_i = 0)} \phi_i} - \mathrm{e}^{-\mathrm{i} \sum_{(i : x_i = 0)}\phi_i}\right)\prod_{i:x_i=0}\cos\left({\frac{\theta_i}{2}}\right)\prod_{i:x_i=1}\sin\left({\frac{\theta_i}{2}}\right)\nonumber\\
        &= (-1)^{\mathrm{wt}(\bar{x})}\left(\cos^2\left(\frac{\gamma}{2}\right)\mathrm{e}^{\mathrm{i} \sum_{(i : x_i = 1)} \phi_i} + \sin^2\left(\frac{\gamma}{2}\right)\mathrm{e}^{-\mathrm{i} \sum_{(i : x_i = 1)}\phi_i}\right)\prod_{i:x_i=0}\sin\left({\frac{\theta_i}{2}}\right)\prod_{i:x_i=1}\cos\left({\frac{\theta_i}{2}}\right)\nonumber\\
        &\quad- 2\mathrm{i}^{n}\sin\left(\frac{\gamma}{2}\right)\cos\left(\frac{\gamma}{2}\right)\sin\left({\sum_{i : x_i = 0}\phi_i}\right)\prod_{i:x_i=0}\cos\left({\frac{\theta_i}{2}}\right)\prod_{i:x_i=1}\sin\left({\frac{\theta_i}{2}}\right).
        \label{eq:a_gamma_theta_phi_x}
    \end{align}
Using the mathematical identity
\[
\cos^2(\alpha) \mathrm{e}^{\mathrm{i}\beta}+\sin^2(\alpha) \mathrm{e}^{-\mathrm{i}\beta}=\cos(\beta)+\mathrm{i}\cos(2\alpha)\sin(\beta),
\]
with $\alpha = \frac{\gamma}{2}$ and $\beta = \sum_{\substack{(i : x_i = 1)}}\phi_i$, we can rewrite the first term in Eq.~\eqref{eq:a_gamma_theta_phi_x} to obtain
    \begin{align}
      a^\gamma_{\theta,\phi}(x)&= (-1)^{\mathrm{wt}(\bar{x})}\left(\cos\left(\sum_{\substack{i : x_i = 1}}\phi_i\right) +\mathrm{i}\cos(\gamma) \sin\left(\sum_{\substack{i : x_i = 1}}\phi_i\right)\right)\prod_{i:x_i=0}\sin\left({\frac{\theta_i}{2}}\right)\prod_{i:x_i=1}\cos\left({\frac{\theta_i}{2}}\right) \nonumber\\
        &\quad- 2\mathrm{i}^{n}\sin\left(\frac{\gamma}{2}\right)\cos\left(\frac{\gamma}{2}\right)\sin\left({\sum_{i : x_i = 0}\phi_i}\right)\prod_{i:x_i=0}\cos\left({\frac{\theta_i}{2}}\right)\prod_{i:x_i=1}\sin\left({\frac{\theta_i}{2}}\right)
        \nonumber\\
        &=(-1)^{\mathrm{wt}(\bar{x})}\cos\left(\sum_{\substack{i : x_i = 1}}\phi_i\right)\prod_{i:x_i=0}\sin\left({\frac{\theta_i}{2}}\right)\prod_{i:x_i=1}\cos\left({\frac{\theta_i}{2}}\right)
        \nonumber\\
        &\quad+ \mathrm{i}(-1)^{\mathrm{wt}(\bar{x})}\cos(\gamma)\sin\left({\sum_{i : x_i = 1}\phi_i}\right)\prod_{i:x_i=0}\sin\left({\frac{\theta_i}{2}}\right)\prod_{i:x_i=1}\cos\left({\frac{\theta_i}{2}}\right)
        \nonumber\\
        &\quad- \mathrm{i}^n\sin(\gamma)\sin\left({\sum_{i : x_i = 0}\phi_i}\right)\prod_{i:x_i=0}\cos\left({\frac{\theta_i}{2}}\right)\prod_{i:x_i=1}\sin\left({\frac{\theta_i}{2}}\right).
    \end{align}
Now, substituting Eq.~\eqref{eq:c}, Eq.~\eqref{eq:s}, and Eq.~\eqref{eq:w} into the expression above, we obtain our final result:
    \begin{align}
        a^\gamma_{\theta, \phi}(x) &= (-1)^{\mathrm{wt}(\bar{x})}c_{\theta, \phi}(x) + \mathrm{i}(-1)^{\mathrm{wt}(\bar{x})}\cos({\gamma})w_{\theta, \phi}(x)-\mathrm{i}^n\sin({\gamma})s_{\theta, \phi}(x).
    \end{align}
\hfill\qed

\section{Proof of Lemma \texorpdfstring{\ref{lemma:prob}}{2}---Deriving the Generalized Probability Mass Function}
\label{app:lemma2}

In this appendix, we prove Lemma~\ref{lemma:prob} by deriving the probability mass function from the probability amplitude function.

\noindent\textit{Proof of Lemma~\ref{lemma:prob}.} \
By taking the squared magnitude of Eq.~\eqref{eq:amp2}, we recover the probability mass function. We consider the cases where $n$ is even and where $n$ is odd separately.
\paragraph{Case 1: $n$ is even.}
When $n$ is even, the magnitude squared of Eq.~\eqref{eq:amp2} is
\begin{equation}
    \begin{aligned}
        p^\gamma_{\theta,\phi}(x)&=\added{\left|(-1)^{\mathrm{wt}(\bar{x})}c_{\theta,\phi}(x)+\mathrm{i}(-1)^{\mathrm{wt}(\bar{x})}\cos(\gamma)w_{\theta,\phi}(x)-(-1)^{\frac{n}{2}}\sin(\gamma)s_{\theta,\phi}(x)\right|^2}\\
        &=\added{\left[(-1)^{\mathrm{wt}(\bar{x})}c_{\theta,\phi}(x)-(-1)^{\frac{n}{2}}\sin(\gamma)s_{\theta,\phi}(x)\right]^2+\left[(-1)^{\mathrm{wt}(\bar{x})}\cos(\gamma)w_{\theta,\phi}(x)\right]^2}\\
        &=\added{c_{\theta,\phi}^2(x)+\sin^2(\gamma)s_{\theta,\phi}^2(x)+\cos^2(\gamma)w_{\theta,\phi}^2(x)-2(-1)^{\mathrm{wt}(\bar{x})+\frac{n}{2}}\sin(\gamma)c_{\theta,\phi}(x)s_{\theta,\phi}(x)}.
    \end{aligned}
    \label{eq:prob-even}
\end{equation}
To finalize our result, we make use of the following property:

\[
(-1)^{\mathrm{wt}(\bar{x})+\frac{n}{2}}=(-1)^{n-\mathrm{wt}(x)+\frac{n}{2}}=(-1)^{2(\frac{n}{2}-\mathrm{wt}(x))}(-1)^{\frac{n}{2}+\mathrm{wt}(x)}=(-1)^{\frac{n}{2}+\mathrm{wt}(x)}.
\]
Hence, our result is rendered as
\begin{equation}
    \begin{aligned}
        p^\gamma_{\theta, \phi}(x)&=c^2_{\theta, \phi}(x)+\sin^2(\gamma)s^2_{\theta, \phi}(x)+\cos^2(\gamma)w^2_{\theta, \phi}(x)-2(-1)^{\frac{n}{2}+\mathrm{wt}(x)}\sin(\gamma)c_{\theta, \phi}(x)s_{\theta, \phi}(x).
    \end{aligned}
\end{equation}
\paragraph{Case 2: $n$ is odd.}
The magnitude squared of Eq.~\eqref{eq:amp2} in this case is
\begin{equation}
    \begin{aligned}
        p^\gamma_{\theta,\phi}(x)&=\left\lvert\left[(-1)^{\mathrm{wt}(\bar{x})}c_{\theta, \phi}(x)\right]+\mathrm{i}\left[(-1)^{\mathrm{wt}(\bar{x})}\cos(\gamma)w_{\theta, \phi}(x)-(-1)^{\frac{n-1}{2}}\sin(\gamma)s_{\theta, \phi}(x)\right]\right\rvert^2\\
        &= c^2_{\theta, \phi}(x)+\left[(-1)^{\mathrm{wt}(\bar{x})}\cos(\gamma)w_{\theta, \phi}(x)-(-1)^{\frac{n-1}{2}}\sin(\gamma)s_{\theta, \phi}(x)\right]^2\\
        &= c^2_{\theta, \phi}(x)+\sin^2(\gamma)s^2_{\theta, \phi}(x)+\cos^2(\gamma)w^2_{\theta, \phi}(x)-2(-1)^{\frac{n-1}{2}+\mathrm{wt}(\bar{x})}\sin(\gamma)\cos(\gamma)w_{\theta,\phi}(x)s_{\theta, \phi}(x).
    \end{aligned}
\end{equation}
To simplify the result, we make use of the following:
\[(-1)^{\frac{n-1}{2}+\mathrm{wt}(\bar{x})}=(-1)^{\frac{n-1}{2}+n-\mathrm{wt}(x)}=(-1)^{\frac{3n-1}{2}-\mathrm{wt}(x)}(-1)^{2\mathrm{wt}(x)+1-n}=(-1)^{\frac{n+1}{2}+\mathrm{wt}(x)}.\]
Our final result then becomes
\begin{equation}
    \begin{aligned}
        p^\gamma_{\theta,\phi}(x)
        &=c^2_{\theta, \phi}(x)+\sin^2(\gamma)s^2_{\theta, \phi}(x)+\cos^2(\gamma)w^2_{\theta, \phi}(x)-2(-1)^{\frac{n+1}{2}+\mathrm{wt}(x)}\sin(\gamma)\cos(\gamma)w_{\theta,\phi}(x)s_{\theta, \phi}(x)\\
        &=c^2_{\theta, \phi}(x)+\sin^2(\gamma)s^2_{\theta, \phi}(x)+\cos^2(\gamma)w^2_{\theta, \phi}(x)-(-1)^{\frac{n+1}{2}+\mathrm{wt}(x)}\sin(2\gamma)w_{\theta,\phi}(x)s_{\theta, \phi}(x).
    \end{aligned}
\end{equation}
Combining the results for the even and odd cases of $n$, the final probability mass function becomes
    \begin{align}
p_{\theta,\phi}^\gamma(x)&=c_{\theta,\phi}^2(x) + \sin^2(\gamma)s^2_{\theta, \phi}(x)+\cos^2(\gamma)w^2_{\theta, \phi}(x)\nonumber\\
        &\quad-s_{\theta, \phi}(x)(-1)^{\mathrm{wt}(x)}
        \begin{cases}
            2(-1)^{\frac{n}{2}}\sin(\gamma)c_{\theta, \phi}(x), &\text{ $n$ even},\\
            (-1)^{\frac{n+1}{2}}\sin(2\gamma)w_{\theta, \phi}(x), &\text{ $n$ odd}.
            \end{cases}
    \end{align}
\hfill\qed

\section{Evaluation of \texorpdfstring{$\mathbf{w_{Q^n}(x)}$}{wQnx}}
\label{app:q1}

In this appendix, we evaluate the quantity $w_{Q^n}(x)$ for the strategy profile $Q^n=(0, \frac{\pi}{n})^n$. The calculation shows that this function vanishes for every bit string $x$:
    \begin{align}
        w_{Q^n}(x)
        &= \sin\left({\sum_{i:x_i=1}\frac{\pi}{n}}\right)\prod_{i:x_i=0}\sin\left({\frac{0}{2}}\right)\prod_{i:x_i=1}\cos\left({\frac{0}{2}}\right) \nonumber\\
        &= \sin\left({\mathrm{wt}(x)\frac{\pi}{n}}\right)\prod_{i:x_i=0}\sin\left({\frac{0}{2}}\right)\prod_{i:x_i=1}\cos\left({\frac{0}{2}}\right)\nonumber\\
        &= \sin\left({\mathrm{wt}(x)\frac{\pi}{n}}\right)\left(0^{n-\mathrm{wt}(x)}\right)\left(1^{\mathrm{wt}(x)}\right)\nonumber\\
        &= \sin\left({\mathrm{wt}(x)\frac{\pi}{n}}\right)\left(0^{n-\mathrm{wt}(x)}\right)\nonumber\\
        &=\sin({\pi})\added{\mathbb{I}}[\mathrm{wt}(x)=n]\nonumber\\
        &=0.
    \end{align}

\section{Evaluation of \texorpdfstring{$\mathbf{w_{\theta, \phi}(x)}$}{wthetaphix} for Player \texorpdfstring{$\mathbf{a}$}{a} for \texorpdfstring{$\mathbf{Q^n}$}{Qn} Strategy Profile}
\label{app:q2}
In this appendix, we evaluate $w_{\theta,\phi}(x)$ when all players except player $a$ use the strategy $Q=(0,\frac{\pi}{n})$. The calculation separates the contribution of player $a$ from that of the remaining players, leading to a piecewise expression that depends on whether player $a$ volunteers.
    \begin{align}
        w_{\theta, \phi}(x)\bigg|_{(\theta_i, \phi_i)=Q, \forall i\neq a}&=\sin\left({\sum_{i:x_i=1}\phi_i}\right)\prod_{i:x_i=0}\sin\left({\frac{\theta_i}{2}}\right)\prod_{i:x_i=1}\cos\left({\frac{\theta_i}{2}}\right)\Bigg|_{(\theta_i, \phi_i)=Q, \forall i\neq a}\nonumber\\
        &=\sin\left({x_a\phi_a+\sum_{\substack{i:x_i=1 \\ i\neq a}} \phi_i}\right)\sin^{(1-x_a)}\left({\frac{\theta_a}{2}}\right)\cos^{x_a}\left({\frac{\theta_a}{2}}\right)\nonumber\\
        &\quad\times\prod_{\substack{i:x_i=0 \\ i\neq a}}\sin\left({\frac{\theta_i}{2}}\right)\prod_{\substack{i:x_i=1 \\ i\neq a}}\cos\left({\frac{\theta_i}{2}}\right)\Bigg|_{(\theta_i, \phi_i)=Q, \forall i\neq a}\nonumber\\
        &=\sin\left({x_a\phi_a+\sum_{\substack{i:x_i=1 \\ i\neq a}} \frac{\pi}{n}}\right)\sin^{(1-x_a)}\left({\frac{\theta_a}{2}}\right)\cos^{x_a}\left({\frac{\theta_a}{2}}\right)\added{\mathbb{I}}[\forall i \neq a, x_i=1]\nonumber\\
        &=\sin\left({x_a\phi_a+\left(\frac{n-1}{n}\right)\pi}\right)\sin^{(1-x_a)}\left({\frac{\theta_a}{2}}\right)\cos^{x_a}\left({\frac{\theta_a}{2}}\right)\added{\mathbb{I}}[\forall i \neq a, x_i=1]\nonumber\\
        &=\sin\left({x_a\phi_a+\frac{\pi n}{n}}-\frac{\pi}{n}\right)\sin^{(1-x_a)}\left({\frac{\theta_a}{2}}\right)\cos^{x_a}\left({\frac{\theta_a}{2}}\right)\added{\mathbb{I}}[\forall i \neq a, x_i=1]\nonumber\\
        &=-\sin\left({x_a\phi_a-\frac{\pi}{n}}\right)\sin^{(1-x_a)}\left({\frac{\theta_a}{2}}\right)\cos^{x_a}\left({\frac{\theta_a}{2}}\right)\added{\mathbb{I}}[\forall i \neq a, x_i=1]\nonumber\\
        &=\begin{cases}
            \sin\left({\frac{\pi}{n}}\right)\sin\left({\frac{\theta_a}{2}}\right), &\text{ for } x_a=0, x_i=1, \forall i \neq a,\\
            \sin\left({\frac{\pi}{n}-\phi_a}\right)\cos\left({\frac{\theta_a}{2}}\right), &\text{ for } x_a=1, x_i=1, \forall i \neq a,\\
            0, &\text{ otherwise}.
        \end{cases}
    \end{align}

\section{Evaluation of \texorpdfstring{$\mathbf{w_{A^n}(x)}$}{wAnx}}
\label{app:a1}
In this appendix, we evaluate $w_{A^n}(x)$ for the strategy profile $A^n=(0,\frac{\pi}{2})^n$. The resulting expression is nonzero only when \added{$n$ is odd and $x=1^n$ is a string of ones. In this case, $\prod_{(i:x_i=0)}0=1$, since the index set of the product operator is empty.}
    \begin{align}
        w_{A^n}(x)&=\sin\left(\sum_{i:x_i=1}\frac{\pi}{2}\right)\prod_{i:x_i=0}\sin\left(\frac{0}{2}\right)\prod_{i:x_i=1}\cos\left(\frac{0}{2}\right)\nonumber\\ 
        &= \sin\left(\mathrm{wt}(x)\frac{\pi}{2}\right)\prod_{i:x_i=0}0\nonumber\\
        &= \sin\left(\mathrm{wt}(x)\frac{\pi}{2}\right)\added{\mathbb{I}}[\forall i:x_i=1]\nonumber\\
        &=\sin\left(\mathrm{wt}(x)\frac{\pi}{2}\right)\added{\mathbb{I}}[\mathrm{wt}(x)=n]\nonumber\\
        &= \sin\left(\frac{n\pi}{2}\right)\added{\mathbb{I}}[\mathrm{wt}(x)=n]\nonumber\\
        &= \added{\mathbb{I}}[\mathrm{wt}(x)=n]
        \begin{cases}
        0, &\text{for even $n$},\\
        1, &\text{for $n \equiv 1 \pmod{4}$},\\
        -1, &\text{for $n \equiv 3 \pmod{4}$},
        \end{cases}\nonumber\\
        &= (-1)^{\frac{n-1}{2}}\added{\mathbb{I}}[n \text{ odd}]\added{\mathbb{I}}[x=1^n].
    \end{align}

\section{Evaluation of \texorpdfstring{$\mathbf{w_{\theta, \phi}(x)}$}{wthetaphix} for Player \texorpdfstring{$\mathbf{a}$}{a} in \texorpdfstring{$\mathbf{A^n}$}{An} Strategy Profile}
\label{app:a2}
In this appendix, we evaluate $w_{\theta,\phi}(x)$ when all players except player $a$ use the strategy $A=(0,\frac{\pi}{2})$.
    \begin{align}
        w_{\theta, \phi}(x)&\bigg|_{(\theta_i, \phi_i)=A, \forall i\neq a}= \sin\left(\sum_{i:x_i=1}\phi_i\right)\prod_{i:x_i=0}\sin\left(\frac{\theta_i}{2}\right)\prod_{i:x_i=1}\cos\left(\frac{\theta_i}{2}\right)\bigg|_{(\theta_i, \phi_i)=A, \forall i\neq a}\nonumber\\
        &=\sin\left(x_a\phi_a+\sum_{\substack{i:x_i=1\\ i\neq a}}\phi_i\right)\sin^{1-x_a}\left(\frac{\theta_a}{2}\right)\cos^{x_a}\left(\frac{\theta_a}{2}\right)\prod_{\substack{i:x_i=0\\i\neq a}}\sin\left(\frac{\theta_i}{2}\right)\prod_{\substack{i:x_i=1\\i\neq a}}\cos\left(\frac{\theta_i}{2}\right)\bigg|_{(\theta_i, \phi_i)=A, \forall i\neq a}\nonumber\\
        &= \sin\left(x_a\phi_a+\sum_{\substack{i:x_i=1\\i\neq a}}\frac{\pi}{2}\right)\sin^{1-x_a}\left(\frac{\theta_a}{2}\right)\cos^{x_a}\left(\frac{\theta_a}{2}\right)\added{\mathbb{I}}[\forall i\neq a, x_i=1]\nonumber\\
        &= \sin\left(x_a\phi_a+(n-1)\frac{\pi}{2}\right)\sin^{1-x_a}\left(\frac{\theta_a}{2}\right)\cos^{x_a}\left(\frac{\theta_a}{2}\right)\added{\mathbb{I}}[\forall i\neq a, x_i=1].
    \end{align}
For the next step, we make use of the following:
    \begin{align}
    \sin\left(x\phi+(n-1)\frac{\pi}{2}\right)&=
    \begin{cases}
    -\cos(x\phi), &\text{ for } n \equiv 0 \pmod{4},\\
    \sin(x\phi), &\text{ for } n \equiv 1 \pmod{4},\\
    \cos(x\phi), &\text{ for } n \equiv 2 \pmod{4},\\
    -\sin(x\phi), &\text{ for } n \equiv 3 \pmod{4},
    \end{cases}\nonumber\\
    &= -(-1)^{\frac{n}{2}}\cos(x\phi)\added{\mathbb{I}}[n \text{ even}]+(-1)^{\frac{n-1}{2}}\sin(x\phi)\added{\mathbb{I}}[n\text{ odd}].
    \end{align}
Hence, we have
    \begin{align}
        w_{\theta, \phi}(x)\bigg|_{(\theta_i,\phi_i)=A, \forall i\neq a}&=\left(-(-1)^{\frac{n}{2}}\cos(x_a\phi_a)\added{\mathbb{I}}[n \text{ even}]+(-1)^{\frac{n-1}{2}}\sin(x_a\phi_a)\added{\mathbb{I}}[n \text{ odd}]\right)
        \nonumber\\
        &\quad\times\sin^{1-x_a}\left(\frac{\theta_a}{2}\right)\cos^{x_a}\left(\frac{\theta_a}{2}\right)\added{\mathbb{I}}[\forall i\neq a, x_i=1].
    \end{align}
\twocolumngrid
\bibliography{bibliography}

\end{document}